\documentclass[11pt]{article}

\usepackage{cite}

\usepackage{amsmath}
\usepackage{algorithm}
\usepackage{algpseudocode}
\usepackage{enumerate}
\usepackage{enumitem}
\usepackage{amsthm}
\usepackage{amssymb}
\usepackage{pgf,tikz}
\usepackage{thmtools}
\usepackage{lscape}
\usepackage{subcaption}
\usepackage{hyperref}
\usepackage{mathrsfs}
\usepackage{enumitem}
\usepackage{multicol}
\usepackage{hyperref}
\usepackage{mathrsfs}

\usepackage[utf8]{inputenc}
\usepackage[T1]{fontenc}
\usepackage{lmodern}
\usepackage{ulem}

\usetikzlibrary{arrows,matrix}
\usetikzlibrary{shapes.geometric}

\newtheorem{teo}{Theorem}[section]

\newtheorem{cor}[teo]{Corollary}
\newtheorem{lem}[teo]{Lemma}
\newtheorem{Def}[teo]{Definition}
\newtheorem{ex}[teo]{Example}

\providecommand{\F}{\ensuremath{\mathbb{F}_{q}} }

\providecommand{\Fn}{\ensuremath{\mathbb{F}_{q}^{n}} }

\providecommand{\gf}[1]{\ensuremath{\mathbb{F}_{#1}}}

\title{$\F[G]$-modules and $G$-invariant codes}

\author{
E. J. Garcia-Claro\thanks{Research project supported by CONACyT (Consejo Nacional de Ciencia y Tecnología, Mexico) Grant 401846.} and H. Tapia-Recillas\\ 
\small{Department of Mathematics}\\
\small{Universidad Autónoma Metropolitana-Iztapalapa, Mexico city, Mexico}\\
\small{eliasjaviergarcia@gmail.com}\\
\small{htr@xanum.uam.mx}}

\begin{document}

\maketitle

\begin{abstract}
If $\F$ is a finite field, $C$ is a vector subspace of $\Fn$ (linear code), and $G$ is a subgroup of the group of linear automorphisms of $\Fn$,  $C$ is said to be $G$-invariant if $g(C)=C$ for all $g\in G$. A solution to the problem of computing all the $G$-invariant linear codes $C$ of $\Fn$  is offered. This will be referred to as the invariance problem. When $n=|G|t$, we determine conditions for the existence of an isomorphism of $\F[G]$-modules between $\Fn$ and $\F[G]\times \cdots \times \F[G]$ ($t$-times), that preserves the Hamming weight. This reduces the invariance problem to the determination of the $\F[G]$-submodules of $\F[G]\times \cdots \times \F[G]$ ($t$-times).  The concept of the Gaussian binomial coefficient for semisimple $\F[G]$-modules, which is useful for counting $G$-invariant codes, is introduced. Finally, a systematic way to compute all the $G$-invariant linear codes $C\subseteq \Fn$ is provided, when $(|G|,q)=1$.
\end{abstract}


%

%
%
%
%

 

\section{Introduction}\label{intro}

Throughout this work $\F$ denotes the finite field with $q$ elements. A code $C$ is a vector  subspace of $\F^{n}$.  The Hamming distance $d(\bf{x}, \bf{y})$ between two vectors $\textbf{x}, \textbf{y} \in \F^n$ is defined to be the number of coordinates in which $\bf{x}$ and $\bf{y}$ differ. The Hamming weight $wt(\bf{x})$
of a vector $\textbf{x} \in \Fn$ is the number of non-zero coordinates in $\bf{x}$. If $C$ is a code, its minimum weight is $wt(C):=\min\{wt(c)\, : \, c\in C\}$. The group of linear automorphisms of $C$ is denoted by $Aut_{\F}(C)$, and the monomial automorphism group  of $C$ is denoted by $MAut(C)$. Given a code, to make emphasis on its dimension $k$ and the length $n$ of its  vectors (codewords), it is usually  referred to as an $[n,k]$-code. Given an $[n,k]$-code, to make emphasis on its minimum weight $\delta$, it is usually  referred to as an $[n,k, \delta]$-code.\\
The group algebra $\F[G]$ of a finite group $G$ over $\F$, is the set of the formal linear combinations of elements in $G$ with coefficients in $\F$, i.e., $\F[G]:=\left\lbrace \sum_{g\in G} a_{g} g \, : \, a_{g}\in \F \right\rbrace$. This set is a ring  with the usual sum of vectors and the multiplication given by extending  the operation of $G$. 
Its identity, denoted by $1$, is the identity element of $G$ times the unity of $\F$. The Hamming weight $wt_{G}(\bf{x})$ of an element $\textbf{x} \in \F[G]$ is the cardinality of its support respect to $G$.\\
One of the most well-known classes of linear codes is the one of cyclic codes; these are precisely the invariant codes under the automorphism $\sigma$ of $\F^n$ given by $\sigma(a_0, a_1,...,a_{n-1}) := (a_{n-1}, a_0, ..., a_{n-2})$   \cite[$Ch.4$]{cod2}. Cyclic codes and some of their generalizations, such as consta-cyclic codes and quasi-cyclic codes, respond to the question of finding the invariant codes under some automorphisms of $\Fn$. For instance, cyclic codes are the invariant codes under the cyclic shift $\sigma$; quasi-cyclic codes are the invariant codes under $\sigma^k$, where $k$ is a divisor of $n$; and consta-cyclic codes are those which are invariant under the mapping $\rho(a_0, a_1,...,a_{n-1}) := (ca_{n-1}, a_0, ..., a_{n-2})$ with $c\in \F^*$ \cite{conscod}. All of these families of codes have in common that they are invariant under the groups $\langle \sigma \rangle$, $\langle\sigma^k \rangle$ and $\langle \rho \rangle$, respectively.\\
In general, given  $G\leq Aut_{\F}(\F^n)$, one might ask which are the codes $C\subset \Fn$ such that $g(C)=C$ for all $g\in G$. Offering an answer to this question is the first aim of this work; the second aim is to develop a formula to count $G$-invariant codes. When $G\leq Aut_{\F}(\F^n)$, $\F^n$ can be endowed with a structure of left $\F[G]$-module where $\left( \sum_{g\in G } \lambda_{g}g \right)\cdot  v:= \sum_{g\in G } \lambda_{g}g(v)$ for all $v\in \Fn$. So,  the $G$-invariant codes of $\Fn$ are precisely its $\F[G]$-submodules. Henceforth, every module will be assumed to be a left module.\\
If $C\subset \Fn$ is a $G$-invariant code for some subgroup $G\leq Aut_{\F}(\Fn)$, then $\{g\restriction_{C} \, : \, g\in G\}\leq Aut_{\F}(C)$. Thus the invariance problem is related to the problem of finding codes which have a non-trivial permutation (monomial) automorphism group. Some authors have addressed this last problem. For example, in \cite{ccagp},  B. K. Dey and B. Sundar Rajan investigated the algebraic structure of a class of codes that they called $G$-invariant codes. In their work,  a $G$-invariant code over $\F$ is a code that is closed under an arbitrary abelian group of permutations with exponent  relatively prime to $q$. They characterized the dual codes  and self-duality of these $G$-invariant codes. Furthermore, they offered a minimum weight  bound for $G$-invariant codes and extended Karlin's decoding algorithm \cite{karlin} from systematic quasi-cyclic codes to systematic quasi-abelian codes. In \cite{cppg}, W. Knapp and P. Schmid consider $[n,k]$-codes  $C\subseteq \Fn$ such that "the permutation part of their monomial automorphisms" given by $MAut_{Pr}(C):=\{f\in S_{n} : df\in MAut(C) \}$ (where $d$ is represented by a diagonal invertible matrix on the canonical basis) contains $A_{n}$, $S_{n}$, or the \textit{Mathieu group}. They proved that  if $n> 6$ and $A_{n}\subseteq MAut_{Pr}(C)$, then $C$ should be equivalent to the zero code, $\F^n$, the repetition code or its dual. Moreover, they classified (up to equivalence) the few exceptions that occurred  when $n\leq 6$ and studied the case in which $MAut_{Pr}(C)$ contains the \textit{Mathieu group} (but not $A_{n}$).\\
To obtain the main results of this work, the $G$-invariant codes are considered as semisimple $\F[G]$-modules (see this definition in Section \ref{stwo}). Taking advantage of the property of semisimplicity, a method to compute them and a formula to count them is developed. Our results apply for arbitrary finite groups, these might be permutation groups or not\footnote{Actually these groups might have elements that do not preserve the Hamming weight, but for some applications of the results to coding theory, these can be regarded as groups of permutations or monomial transformations.}.\\
This work addressed different questions related to the invariance problem. It is organized as follows: In Section \ref{stwo}, some preliminaries about semisimple modules and some results that will be used later are presented. Then, in Section \ref{sthree}, we address the question of determining when $\Fn$, with a structure of $\F[G]$-module as above, is isomorphic, with an isomorphism that preserves the Hamming weight, to a direct sum of copies of  $\F[G]$. A  particular clear  answer to this, but not unique, occurs when $G$ is the group generated by the cyclic shift of $\Fn$. Throughout Sections \ref{sfive} and \ref{ssix}, we introduce and study the concept of   the \textit{Gaussian binomial coefficient} for finite semisimple $\F[G]$-modules, and develop an algorithm to efficiently compute all the possible sums of a collection of simple isomorphic $\F[G]$-submodules of a given finite $\F[G]$-module. This algorithm is used later in Section \ref{sseven} to provide a method to determine all the $G$-invariant codes of $\Fn$ when $(|G|,q)=1$. Finally, in Sections \ref{seight} and \ref{snine}, we propose a theoretically possible solution to solve the invariance problem when our method could not be applied, and give some examples illustrating the most important results, respectively.

\section{Preliminaries}\label{stwo}

As it was mentioned above, the $G$-invariant codes in $\F^n$ are precisely its  $\F[G]$-submodules. We use this module structure they possess to compute them and to count them. For that, these modules are required to be semisimple. In this section  we present the concept of a semisimple module among some other essential definitions and  results that will be used in the upcoming sessions. For reviewing properties about semisimple modules and rings see \cite[Section $3B$]{rep3} and \cite[Section 25]{rep2}. \\
Let $A$  be a ring, $M$ and $S$ be  $A$-modules. $S$ is a \textbf{simple module} if $S$ does not have  proper submodules different from $\textbf{0}$. $M$ is a  \textbf{semisimple module} if it can be expressed as a direct sum of simple submodules; or equivalently,  if for any $A$-submodule $N$ of $M$ there exists a submodule  $U$ of $M$ such that $M= N\oplus U$. $A$ is a \textbf{semisimple ring} if  every non-zero $A$-module is semisimple.\\
Two elements $x,y\in A$ are called \textbf{orthogonal} if $xy=yx=0$. An element $e\in A$ is called \textbf{idempotent} if $e^2=e$; and it is \textbf{primitive} if $e=f+g$, where $f,g\in A$ are orthogonal idempotents,  implies $f=0$ or $g=0$.\\
Let $M$ and $N$ be  $A$-modules. If there exists an $A$-module $U$ such that $M\cong N\oplus U$, it is said that $N$ \textbf{divides} $M$, and denoted by $N\mid M$. If $M$ is a semisimple finitely generated $A$-module, and $S$ is a simple $A$-module such that $S\mid M$, the \textbf{multiplicity} $n$ of $S$ in $M$ is defined as the greatest natural number such that $nS\mid M$ with $nS=S\oplus \cdots \oplus S$ ($n$-times). Let $I\subset A$ be an ideal. If I is a simple $A$-submodule, it is called minimal ideal. Let $M$ be a finitely generated $A$-module over a semisimple ring. If $\{Af_{1},...,$
  $ Af_{r}\}$ is a collection of minimal ideals of $A$ such that $Af_{j}\mid M$ for all $j\in \{1,...,r\}$ and for any simple $A$-submodule $N$ of $M$ there exists a unique $j\in \{1,...,r\}$ such that $N\cong Af_{j}$, then $\{Af_{1},..., Af_{r}\}$ will be said to be a \textbf{basic set of ideals}  for $M$. A collection of idempotents $\{e_{1},...,e_{r}\}\subset A$   such that $\{Ae_{1},..., Ae_{r}\}$ is a basic set of ideals  for $M$ will be called a \textbf{ basic set of idempotents } for $M$. If $e\in A$ is a primitive idempotent such that $Ae\mid M$, then the \textbf{homogeneous component associated with $e$ ($Ae$)} will be the $A$-submodule of $M$ defined by $\sum_{U\leq M \; \wedge U\cong Ae}U$\footnote{Similarly, it will be said that $e$ ($Ae$) is associated with $H$.}.\\
Schur's lemma \cite[Lemma 2.6.14]{grouprings}  and Maschke's theorem \cite[Theorem 3.4.7]{grouprings} are well-known results that could be found in many books of Algebra. However, these are usually  presented in different  contexts and ways. For that reason, and to make easier the reading this work, we present them below in a context that is necessary for our applications.

\begin{lem}[Schur's lemma]\label{schur} Let $M$ and $N$ be simple $\F[G]$-modules. Let $f :M\rightarrow N$ be a non-zero morphism. Then $f$ is an isomorphism.
\end{lem}

\begin{teo}[Maschke's theorem]\label{maschke} 
Let $G$ be a finite group. Then, the group algebra $\F[G]$ is semisimple if and only if $|G|$ is invertible in $\F$.
\end{teo}

 Theorem  \ref{maschke} and \cite[Theorem 25.10]{rep2}  imply the following result: If $G$ is a finite group such that $(|G|,q)=1$ and $0\neq M$ is a finite $\F[G]$-module, then there exists  a collection $\{I_{j}\}_{j=1}^{t}$ of minimal ideals of $\F[G]$ such that $M\cong \oplus_{j=1}^{t} I_{j}$. This result gives some light on how to compute the $G$-invariant codes in $\F^n$. First, all the $\F[G]$-submodules ($G$-invariant codes) of $\F^n$ isomorphic to minimal ideals (i.e., simple $\F[G]$-modules) should be computed (a solution to this is presented in Section \ref{steps}). Then, all the possible direct sums of these modules should be determined (a solution to this is presented in Section \ref{ssix}).\\
The following lemma is just a slight modification of \cite[Theorem 4.3]{nagao} (part $i$), in the context of $\F[G]$-modules.

\begin{lem}\label{iso-ev}
 Let $A=\F[G]$. $e\in A$ be a primitive idempotent,  $I=Ae$, and $M$ an $A$-module. Then, 
$\eta_{e}: Hom_{A}(Ae,M)\mapsto eM$ given by $\varphi\mapsto\varphi(e)$ is an isomorphism of $\F$-vector spaces.
\end{lem}

Note that a basic set of idempotents is a collection of primitive orthogonal idempotents. Its elements are primitive because they generate minimal ideals, and are orthogonal by Lemmas \ref{iso-ev} and \ref{schur}.\\ 
In the following example the decomposition of $\gf{5}[S_{3}]$ into its homogeneous components is computed. This will be used  later again in Example \ref{ex-nocom}. 

\begin{ex}\label{ex1}
Let $S_{3}=\langle a,b \mid a^3=b^2=1, \, bab^{-1}=a^2 \rangle=\{1,a,a^2,b,ba,ba^2\}$ and $A=\gf{5}[S_{3}]$. As $(|S_{3}|, 5)=1$, $A$ is semisimple (by Theorem \ref{maschke}). Then $\{e_{1}=1+a+a^2+b+ba+ba^2, e_{2}=1+a+a^2 +4b+4ba+4ba^2, e_{3}=2+3a^2 +2b +3ba^2, e_{4}=2+3a+3b+2ba^2 \}$ is a set of  primitive orthogonal idempotents of $A$ such that $1=\sum_{i=1}^{4}e_{i}$ \footnote{This idempotents were computed using  \cite[Theorem 3.3]{dfp}  and doing some computations in Sage.}. By Lemmas \ref{iso-ev} and \ref{schur}, $Ae_{1}\ncong Ae_{2}$  because $e_{1}A e_{2}=e_{1}\gf{5} e_{2}=\gf{5}e_{1}e_{2}=0 $. Besides, $ Ae_{3}\cong Ae_{4}$ (by \cite[Proposition 1.3]{ako}) because  $y=1+2a+b$ is a unity in $A$ (with inverse $2+4a+4b+2ba+2ba^2$) such that $ye_{3}y^{-1}=e_{4}$. Therefore, $H_{1}=Ae_{1}$, $H_{2}= Ae_{2}$ and $H_{3}=Ae_{3}\oplus Ae_{4}$ are the homogeneous components of $A$.
\end{ex}  


From now on,
\begin{itemize}
    \item[$\bullet$]  $A$ will denote the semisimple group algebra of a finite group $G$ over  the finite field $\F$ (when $(|G|,q)=1$), unless stated otherwise.
    \item[$\bullet$] If $M$ and $N$ are modules over a ring, $N\leq M $ will denote that $N$ is a submodule of $M$.   
\end{itemize}

\begin{lem}\label{ci-simp}
 Let $I\leq A$ be a minimal ideal such that $I$ has multiplicity $1$ in $A$ and $n\in \mathbb{Z}^+$. If $N\neq 0$ is a cyclic $A$-submodule of $nI$, then $N\cong I$.
\end{lem}

\begin{proof}
Let $N\neq 0$ be a cyclic $A$-submodule of $nI$. Then there exists $0\neq x\in nI$ such that $N=Ax$ and thus $N$ is isomorphic to $A/ann(x)$ (where $ann(x)$ is the annihilator of $x$). So $N$ is isomorphic to a direct complement $S$ of $ann(x)$ in $A$, which there exists because $A$ is semisimple, and hence $N\cong S$. On the other hand, the unique simple $A$-submodule (up to isomorphism)  that divides $nI$ is $I$ and thus the unique simple $A$-submodule (up to isomorphism)  that divides $N$ is $I$. Consequently, there exists $k\in \mathbb{Z}^+$ with $k\leq n$ such that $kI\cong N\cong S\leq A$ and as the multiplicity of $I$ in $A$ is $1$, then $k=1$ and so $I\cong N$.
\end{proof}

\section{Isomorfism of \Fn with $\F[G]\times \cdots \times \F[G]$}\label{sthree}

Let $C_n$ denote the cyclic group of order $n$ generated by the cyclic shift $\sigma$. This group acts by evaluation on $\Fn$  endowing it with a structure of $\F[C_n]$-module. Another natural $\F[C_{n}]$-module is $\F[C_{n}]$ itself. A classic way of determining the cyclic linear codes of $\Fn$ is by using the  isomorphism of $\F[C_n]$-modules,  $\phi: \Fn \rightarrow \F[C_n]$  given by $(a_0,a_1,...,a_{n-1}) \mapsto  a_0+ a_1\sigma+...+a_{n-1}\sigma^{n-1}$, this provides a bijection between the ideals of the group algebra $\F[C_n] $ and the invariant codes of $\Fn$ under the cyclic shift $\sigma$. Furthermore, this $\phi$  preserves the Hamming weight, i. e.,  for all $v\in \Fn$, $wt(v)= wt_{C_n}(\phi(v))$. Let us take that situation to a more general context. If $G\leq Aut_{\F}(\Fn)$ and $n=t|G|$, one may ask whether  there exists an isomorphism of $\F[G]$-modules between $\Fn$ and $\oplus_{i=1}^{t}\F[G]$ preserving Hamming weight. The following result will help to answer that question.\\
Let $G$ be a finite group. If $M$ is an $\F[G]$-module, the representation induced by left multiplication by elements of $G$ in $M$ is $\rho: G\rightarrow Aut_{\F}(M)$ given by $g\mapsto \rho(g)$, where $\rho(g)(v)=gv$ for all $v\in M$.  If $R=\oplus_{i=1}^{t}\F[G]=\F[G]\times \cdots \times \F[G]$ ($t$-times) is an external direct sum of $\F[G]$ (i.e. with pointwise  addition  and pointwise  left multiplication by elements of $\F[G]$), the Hamming weight is defined on $R$ as $wt'((i_{1},...,i_{t}))=\sum_{j=1}^{t} wt_{G}(i_{j})$ for all $(i_{1},...,i_{t})\in R$.

\begin{teo}\label{isoe}
Let $G\leq Aut_{\F}(\Fn)$ with $n=|G|\cdot t$. Let $\pi_{j}:R\rightarrow \F[G]$ be  given by $(i_{1},...,i_{t})\mapsto i_{j}$ for $j=1,...,t$, where $R= \F[G]\times \cdots \times
  \F[G]$ ($t$-times); let  $\mu$ be the canonical basis of $\Fn$, and $\eta=\cup_{j=1}^{t}G_{j}$ where $G_{j}=\{v\in R \,: \, \pi_{j}(v)\in G\, and \, \pi_{i}(v)=0,\; for\; i\neq j\}$. Let $\rho:G\rightarrow Aut_{\F}(\Fn)$, $\rho':G\rightarrow Aut_{\F}(R)$ be the representations induced by the $\F[G]$-modules structure of $\Fn$ (given by evaluation) and $R$ (given by left multiplication), respectively. Let $g\in G$, $[\rho(g)]_{\mu}$ be the matrix of $\rho(g)$ with respect to the basis $\mu$, and $[\rho'(g)]_{\eta}$ the matrix of $\rho'(g)$ with respect to the basis $\eta$. Let $S=\{s_{1},...,s_{k}\}$ be a generating set for $G$, and $M\leq GL(n,q)$ the group of monomial matrices. Then, there is a bijection $\phi$ from $H:=\{A\in M\, : \, [\rho(s_{i})]_{\mu}=A^{-1}[\rho'(s_{i})]_{\eta}A,\, i=1,...,k \}$ to the set $L$ of the $\F[G]$-isomorphisms from $\Fn$ to $R$ that preserve the Hamming weight given by $A\mapsto f_{A}$ with $[f_{A}(v)]_{\eta}:=A[v]_{\mu}$ for all $v\in \Fn$.
\end{teo}

\begin{proof} 
Let $wt$ and $wt'$ denote the Hamming weight on $\Fn$ and $R$, respectively. By definition $wt' (r):=\sum_{j=1}^{t}wt_{G}(\pi_{j}(r))$. Let $h:L\rightarrow H$ be given by $f\mapsto _{\eta}[f]_{\mu}$, for all $f \in L$.\\
First, we show that $\phi$ has, in fact, $L$ as its codomain.  Let $A\in H$,  then  $A[\rho(s_{i})]_{\mu}=[\rho'(s_{i})]_{\eta} A$ for $i=1,...,k$. Let $f:\Fn \rightarrow R$ be the linear transformation given by $[f(v)]_{\eta}:=A[v]_{\mu}$, then $_{\eta}[f]_{\mu}=A$ so that $f(\rho(s_{i}))=\rho'(s_{i})(f)$ for $i=1,...,k$, and thus $\phi(A)=f$ is an isomorphism of $\F[G]$-modules. Now we observe that $f$ preserves the Hamming weight. For $v\in \Fn$, $[f(v)]_{\eta}=A[v]_{\mu}$ and hence $wt(v)=wt(A[v]_{\mu})$ because $A$ is a monomial matrix. But  $wt([f(v)]_{\eta})=\sum_{j=1}^{t} wt([\pi_{j}(f(v))]_{G})=\sum_{j=1}^{t} wt_{G}(\pi_{j}(f(v)))=wt'(f(v))$, by definition of $wt'$, and so $wt(v)=wt'(f(v))$.\\
Next, it is shown that $H$ is, in fact, the codomain of $h$. Let $f\in L$ and $g\in G$, then $f(\rho (g))=\rho' (g)(f)$ and thus $_{\eta}[f\circ\rho (g)]_{\mu}=_{\eta}[\rho' (g)\circ f]_{\mu}$. Let $A:= _{\eta}[f]_{\mu}$, then $[\rho(g)]_{\mu}=A^{-1}[\rho'(g)]_{\eta} A$. It remains to show that $A$ is a monomial matrix. Let $e_{i}$ be the $i$-th canonical vector of $\mu$, then  $wt(Ae_{i})=wt(A[e_{i}]_{\mu})=wt([f(e_{i})]_{\eta})=w
t'(f(e_{i}))=wt(e_{i})=1$,
where the penultimate equality  follows from the fact that $f$ preserves the Hamming weight. Thus $Ae_{i}=c_{i}e_{j_{i}}$ with $j_{i}\in \{1,...,n\}$ and $c_{i}\neq 0$. Hence the $i$-th column of $A$ is $c_{i}e_{j_{i}}$. Furthermore, $A$ is invertible. So if $k,l \in \{1,...,n\}$ and $k\neq l$, then $j_{k}\neq j_{l}$,  and so there exists $\tau\in S_{n}$ such that $Ae_{i}=c_{i}e_{\tau(i)}$ for all $i$.  As it is clear that $\phi$ and $h$ are mutually inverse, the proof is complete.
\end{proof}

\begin{ex}
Let $\alpha$ be the automorphism of $\gf{3}^4$ given by $\alpha(a_{0},a_{1},a_{2},a_{3}):=(a_{1},a_{0},a_{3},a_{2})$, and $G:=\langle \alpha \rangle=\{1, \alpha\}$. $G$ has order 2. Thus one might ask if $\gf{3}^4$  and $R=\gf{3}[G]\times \gf{3}[G]$ are isomorphic as $\gf{3}[G]$-modules with an isomorphism that preserves Hamming weight. Let $\rho:G \rightarrow Aut_{\gf{3}}(\gf{3}^{4})$ be the inclusion and $\rho':G \rightarrow Aut_{\gf{3}}(R)$ be given by $\alpha\mapsto l_{\alpha}$ where $l_{\alpha}$ is the left multiplication by $\alpha$. Let $\mu$ denote the canonical basis of $\gf{3}^4$ and $\eta=\{(1,0),(\alpha,0),(0,1),(0, \alpha)\}$. If we consider the monomial matrix  
$$M:= \begin{bmatrix}
1&0&0&0\\
0&0&1&0\\
0&0&0&1\\
0&1&0&0
\end{bmatrix}$$

then $M[\rho(\alpha)]_{\mu}=[\rho'(\alpha)]_{\eta}M$. Hence, according to Theorem \ref{isoe}, $f:\gf{3}^4\rightarrow R$ given by $[f(v)]_{\eta}:=M[v]_{\mu}$ is an isomorphism of $\gf{3}[G]$-modules that preserves the Hamming weight. Consequently,  $C\subseteq \gf{3}^4$ is a $G$-invariant code if and only if $f(C)$ is a  $\gf{3}[G]$-submodule ($2$-quasi-cyclic code) of $R$ which is isomorphic to $C$ as   $\gf{3}[G]$-module and as metric space, i.e., the $G$-invariant codes  in $\gf{3}^4$ are equivalent to the $2$-quasi-cyclic codes. 

\end{ex} 


\section{The Gaussian binomial coefficient for semisimple $\F[G]$-modules}\label{sfive} 

In this section, a Gaussian binomial coefficient for finite $A$-modules is introduced, and some of its properties are studied. This coefficient will be useful for counting $G$-invariant codes when $(|G|,q)=1$.

\begin{Def}[Gaussian binomial Coefficient] Let $N$ and $M$ be finite $A$-modules. The Gaussian binomial coefficient of $M$ in $N$ is defined as $$\binom{M}{N}_{q}:= |\{U\leq M \; : \; U\cong N \}|.$$\end{Def}

\begin{lem}\label{same-bino}Let $M$, $N$, and $T$ be finite $A$-modules. Let $I\leq A $ be a minimal ideal such that $I\mid M$. Then the following hold:

\begin{enumerate}
\item $\binom{M}{N}_{q}\neq 0$ if and only if $N \mid M$.
\item If $N\mid M$, then $\binom{N}{I}_{q}\leq \binom{M}{I}_{q}$.
\item If $U$ is an $A$-module such that $\binom{kI}{I}_{q}\leq \binom{U}{I}_{q}$, then $kI\mid U$.
\item If $Hom_{A}(T,N)=0$, then  $\binom{M\oplus N}{T}_{q}=\binom{M}{T}_{q}$. 
\item If $Hom_{A}(T,N)=0$, then $\binom{M}{T\oplus N}_{q}=\binom{M}{T}_{q}\binom{M}{N}_{q}$.

\end{enumerate}
\end{lem}

\begin{proof}
\begin{enumerate}
\item $N\mid M$ if and only if $N$ is embedded in $M$,  which occurs if and only if $\binom{M}{N}_{q}\neq 0$.

\item It is clear.

\item If $X$ is a set of simple $A$-submodules of $U$ isomorphic to $I$ such that $\binom{kI}{I}_{q}\leq |X|$, then $\sum_{x\in X}x\cong rI$ for some $r\in \mathbb{Z}^+$. Thus $\binom{kI}{I}_{q}\leq |X| \leq \binom{rI}{I}_{q}$, implying that   $k\leq r$ and so $kI\mid rI \mid U$.

\item $Hom_{A}(T,N)=0$ if and only if $T$ and $N$ have no simple common divisors. Let $M'\cong M$ and $N'\cong N$ be such that the internal direct sum $M'\oplus N'= M\oplus N$. As $M\mid M\oplus N$, then $\binom{M}{T}_{q}=\binom{M'}{T}_{q}\leq \binom{M\oplus N}{T}_{q}=\binom{M'\oplus N'}{T}_{q}$. If the equality does not hold, there exists $L\leq M'\oplus N'$, with $L\cong T$ and $L\nleq M'$. Thus there exists a simple $A$-module $S\leq L$ such that $S\cap M'=0$. Hence $S\mid  N$, but $S\mid T$, which contradicts $Hom_{A}(T,N)=0$.

\item  Consider the function $f: D_{0}:=\{L\leq M\; : \; L\cong T\oplus N\}\longrightarrow D_{1}:= \{U\leq M\; : \; U\cong T\}\times \{Z\leq M\; : \; Z\cong  N\}$ given by $f(L)=(L_1,L_2)$ where $L=L_1\oplus L_2$, $L_1\cong T$, and $L_2\cong N$. $f$ is well defined. Otherwise, there exists $L$ in its domain such that $L=L_1\oplus L_2= X_{1}\oplus X_{2}$ where $L_1\cong X_{1} \cong T$, $L_2\cong X_{2} \cong N$, and $L_1\neq X_{1}$ or $L_2\neq X_{2}$. Without loss of generality, $L_1\neq X_{1}$ and there exists a simple $A$-submodule $S$ of $L_{1}$ that is not contained in $X_{1}$. Then, the multiplicity of $S$ in $S\oplus X_{1}\leq L\cong T\oplus N$ is greater than its multiplicity in $X_{1}\cong T$ and thus $S\mid N$, contradicting $Hom_{A}(T,N)=0$. By a similar argument, $h:D_{1}\rightarrow D_{0}$ given by $h(B,C)=B\oplus C$ is well defined so that $h$ is the inverse of $f$. Therefore, $\binom{M}{T\oplus N}_{q}=\binom{M}{T}_{q}\binom{M}{N}_{q}$. 
\end{enumerate}
\end{proof}

\begin{cor}\label{bcoef-mult} Let $M\cong \oplus_{j=1}^{t} n_{j}I_{j}$  be an $A$-module, where $I_{j}\leq A $ is a  minimal ideal and $n_{j}$ its multiplicity in $M$ for $j=1,...,t$. Let $N\leq M$, if $N\cong \oplus_{l\in J} k_{l}I_{l} $ where $J \subseteq \{1,...,t\}$, then

$$\binom{M}{N}_{q} = \prod_{l \in J} \binom{n_{l}I_{l}}{k_{l}I_{l}}_{q}.$$
\end{cor}

\begin{proof}
Let $J \subseteq \{1,...,t\}$ be such that $N\cong \oplus_{l\in J} k_{l}I_{l} $. As $N\leq M $, then $k_{l}\leq n_{l}$ for all $l\in J$  and thus 

\begin{eqnarray*}
\binom{M}{N}_{q}&=&\binom{\oplus_{j=1}^{t} n_{j}I_{j}}{\oplus_{l\in J} k_{l}I_{l}}_{q}\\
                &=&\prod_{l \in J} \binom{ \oplus_{j=1}^{t} n_{j}I_{j}}{k_{l}I_{l}}_{q} \; \text{by Lemma}\; \ref{same-bino} \; \text{(part $5$)}\\
                &=&\prod_{l \in J} \binom{n_{l}I_{l}}{k_{l}I_{l}}_{q} \; \text{by Lemma}\; \ref{same-bino} \; \text{(part $4$)}.\\               
\end{eqnarray*}

\end{proof}

\begin{Def}
Let $M$ be a finite $A$-module and let  $SS(M)$ be defined as the collection of all simple $A$-submodules of $M$.\\
If $X\subseteq SS(M)$ is such that $\sum_{x\in X}x=M$, then it will be said that $X$ \textit{generates} $M$, or $X$ \textit{is a generating set by simple} $A$-\textit{submodules of} $M$.\\ 
If $Y\subseteq SS(M)$ is such that $\sum_{y\in Y}y=\oplus_{y\in Y}y$, then it will be said that $Y$ is independent, or $Y$ is an \textit{ independent set of simple} $A$-\textit{submodules of} $M$.\\ 
If $X \subseteq SS(M)$ is independent and generates $M$, it will be said that $X$ is a basis for $M$, or $X$ is a \textit{basis by simple} $A$-\textit{submodules of} $M$.
\end{Def}

\begin{lem}\label{basisp} Let $M$ be a finite semisimple $A$-module. 
\begin{enumerate}
\item If $X\subseteq SS(M)$ and $X$ generates $M$, then $X$ contains a basis.

\item If $X\subseteq SS(M)$ is independent and has the cardinality of a basis for $M$, then $X$ is a basis.
\end{enumerate}
\end{lem}

\begin{proof}
\begin{enumerate}

\item Let $X\subseteq SS(M)$ be such that $\sum_{x\in X}x=M$. If $\sum_{x\in X}x\neq \oplus_{x\in X}x$ (i.e. $X$ is not independent), then there exists $x_{0}\in X$ such that $x_{0}\subseteq \sum_{x\in X-\{x_{0}\}}x=M$. If $ X_{0}=X-\{x_{0}\}$ is independent, the proof concludes. Otherwise, the same reasoning can be applied to $X_{0}$, and this process can be repeated until  obtaining a subset of $X$ that is independent and generates $M$.

\item Let $X, Y\subseteq SS(M)$ be such that $\sum_{x\in X}x=\oplus_{x\in X}x$, $M=\oplus_{y\in Y}y$, and $|X|=|Y|$. As $\oplus_{x\in X}x\subset M=\oplus_{y\in Y}y$,  $\oplus_{x\in X}x\cong \oplus_{y\in J}y$ with $J\subseteq Y$, and so  by Krull-Schmidt Theorem \cite[ p. 538]{ama}, $|J|=|X|=|Y|$, implying that  $\oplus_{x\in X}x=\oplus_{y\in Y}y=M$.

\end{enumerate}
\end{proof}

\begin{lem}\label{bcoef} If  $I\leq A$ is a minimal ideal, and $k,n\in \mathbb{Z}^{+}$ with $2\leq k\leq n$, then

 $\binom{nI}{kI}_q=\frac{\binom{nI}{I}_q \left[ \binom{nI}{I}_q -1 \right]\left[\binom{nI}{I}_q -\binom{2I}{I}_q\right]\cdots \left[\binom{nI}{I}_q -\binom{(k-1)I}{I}_q\right]}{\binom{kI}{I}_q \left[ \binom{kI}{I}_q -1 \right]\left[\binom{kI}{I}_q -\binom{2I}{I}_q\right]\cdots \left[\binom{kI}{I}_q -\binom{(k-1)I}{I}_q\right]}.$\\


\end{lem}

\begin{proof}
Let $X=\{N\leq nI \, : \, N\cong kI\}$, and $Y=\{A\subseteq SS(nI)\, : \, |A|=k\, \text{and} \, A \, \text{is independent}\}$ the collection of basis by simple $A$-submodules of the elements of $X$ (by Lemma \ref{basisp}, part $2$). Let $R$ be the equivalence relation on $Y$ given by $aRb$ if and only if $a$  generates the same $A$-module as $b$. If $Y/R$ is the quotient set determined by $R$, then the function $f: Y/R\longrightarrow X$ given by $[x]\mapsto \sum_{m\in [x] }m$ is  a bijection. So $\lvert Y/R \rvert=\lvert X \rvert=\binom{nI}{kI}_q$. On the other hand, every equivalence class in $Y/R$ has the same cardinality, which is the number  of basis by simple $A$-submodules of $kI$. Thus $\lvert Y/R \rvert= \frac{|Y|}{|[x_0]|}$ for some $x_{0}\in Y$ and hence $\binom{nI}{kI}_q=\frac{\text{ \# of l. i. sets of simple} \,A\text{-submodules of}\, nI \,\text{with size} \,k}{ \text{ \# of basis by simple}\, A\text{-submodules of}\, kI}$.\\

 For building an independent set of simple $A$-submodules of $nI$ having cardinality $k$ (i.e., an element of $Y$), we should start by taking a simple $A$-module $S_{1}$, and for that, we have $\binom{nI}{I}_q$ possible choices. For taking another simple $A$-module $S_{2}$ such that the collection $\{S_1, S_2\}$  remains independent, we have $\left[ \binom{nI}{I}_q -\binom{I}{I}_q \right]$ possible choices. To take a simple $A$-module $S_{3}$ such that the collection $\{S_1, S_2, S_3\}$ remains independent, we have $\left[ \binom{nI}{I}_q -\binom{2I}{I}_q \right]$ possible elections. In general, to chose a simple $A$-module $S_{k}$ such that the collection $L=\{S_1, S_2, S_3,...,S_{k-1},S_k\}$ remains independent, we have $\left[ \binom{nI}{I}_q -\binom{(k-1)I}{I}_q \right]$ possibilities. Therefore, at the end of this process, we will construct an independent set of simple  $A$-submodules of $nI$ having size $k$. We did the choices of the $A$-modules $S_{i}$ without worrying about the order, and depending on that, the same set $L$ can be obtained, so that $|Y|=\frac{\binom{nI}{I}_q \left[ \binom{nI}{I}_q -\binom{I}{I}_q \right]\left[\binom{nI}{I}_q -\binom{2I}{I}_q\right]\cdots \left[\binom{nI}{I}_q -\binom{(k-1)I}{I}_q\right]}{k!}.$
 
To build a basis by simple $A$-submodules of $kI$, we  could apply the same reasoning used before. For that, it should be taken into account that an independent set of simple $A$-submodules of $kI$ that has cardinality $k$ is, in fact, a basis for $kI$ (by Lemma \ref{basisp}, part $2$) and hence  $|[x_0]|=\frac{\binom{kI}{I}_q \left[ \binom{kI}{I}_q -\binom{I}{I}_q \right]\left[\binom{kI}{I}_q -\binom{2I}{I}_q\right]\cdots \left[\binom{kI}{I}_q -\binom{(k-1)I}{I}_q\right]}{k!}.$  

\end{proof}


\begin{lem}\label{simp-part} Let $I\leq A $ be a minimal ideal with multiplicity $1$ in $A$, and $n\in \mathbb{Z}^+$. If $U^{*}:=U-\{0\}$ for all $U\leq nI$, then $X:=\{U^{*}\; : \; U\leq nI \; \land \; U\cong I\}$ is a partition of $(nI)^{*}$.
\end{lem}

\begin{proof} If $n=1$, the assertion is true. Otherwise, if $V,W\in X$ and $W\neq V$, then there exist $U$ and $T$ distinct simple $A$-submodules of $nI$ such that $V=U^{*}$ and $W=T^{*}$. As $U\cap T=\{0\}$, $\emptyset= U^{*}\cap T^{*}=V\cap W$. If $x\in (nI)^{*}$, $Ax\subset nI$ and thus $Ax\cong I$ (by Lemma \ref{ci-simp}) so that  $x\in Ax-\{0\}\in X$.\end{proof}

\begin{cor}\label{simp-part-cor}Let $I$ and $n$ be as in Lemma \ref{simp-part}. If $dim_{\F}(I)=k$, then $\binom{nI}{I}_{q}=\frac{q^{nk}-1}{q^{k}-1}.$

\end{cor}

\begin{proof} Let $X$ be as in Lemma \ref{simp-part}, then $(nI)^{*}=\sqcup_{V\in X}V$  and so $|(nI)^{*}|=|X|\cdot |I^{*}|$. Thus  $$\binom{nI}{I}_{q}=|X|=\frac{|(nI)^{*}|}{|I^{*}|}=\frac{|nI|-1}{|I|-1}=\frac{q^{nk}-1}{q^{k}-1}.$$ 
\end{proof}

Observe that Corollary \ref{bcoef-mult} presents the Gaussian binomial coefficient as a product of simpler Gaussian binomial coefficients, which in turn are later expressed in even simpler terms in Lemma \ref{bcoef}. These last terms are finally calculated, when the minimal ideals that appear in them have multiplicity $1$ in their group algebra, in Corollary \ref{simp-part-cor}. As every minimal ideal $I$ of a semisimple commutative group algebra  $\F[G]$ has multiplicity $1$ in $\F[G]$, now we can compute any Gaussian binomial coefficient when $G$ is abelian.

\subsection{Counting all the $G$-invariant codes}

The following result plays an important role in the solution of the invariance problem.

\begin{lem}\label{cont-inv}
Let $G\leq Aut_{\F}(\Fn)$ be such that $A=\F[G]$ is  semisimple. Let  $\{I_{j}\, : \, j=1,...,r\}$ be a basic set of ideals for $\Fn$, and  $H_{j}\cong n_{j}I_{j}$ be the  homogeneous component of $\Fn$  associated with $I_{j}$ for $j=1,...,r$.  Let $S(M):=\{C\subseteq M \, : \, C \, \text{is\, a \,}  G-invariant \, code\}$ for all $G$-invariant code $M\subseteq \Fn$ and $Z:=\{\oplus_{i=1}^{r}M_{i}\, : \, M_{i} \; \text{is\, a \,}  A-submodule \, of\, H_{i},\, for\, all\, i=\{1,...,r \}$. Let $D(B):=\prod_{j\in B}\{1,..., n_{j}\}$ for all $B\in T:= 2^{\{1,...,r\}}-\{\emptyset\}$. Then,

\begin{enumerate}

\item $S(\Fn)=Z$ and $|S(\Fn)|=\prod_{j=1}^{r}|S(H_{j})|$.\\

\item $ |S(\Fn)|=\left[ \sum_{B\in T}\sum_{(t_{j})_{j}\in D(B)}\left(  \prod_{j\in B  } \binom{n_{j}I_{j}}{t_{j}I_{j}}_{q}  \right)\right] + 1$ and $\prod_{j=1}^{r}|S(H_{j})|=\prod_{j=1}^{r}\left( \sum_{t_{j}=1 }^{n_{j}} \binom{n_{j}I_{j}}{t_{j}I_{j}}_{q} + 1 \right)$.

\end{enumerate}
\end{lem}

\begin{proof}

\begin{enumerate}

\item It is clear that $ S(\Fn)=Z$. Let $f: S(\Fn)\longrightarrow \prod_{j=1}^{r} S(H_{j})$ be given by $N\mapsto (N_{j})_{j=1}^{r}$, where $N_{j}$ is the homogeneous component of $N$ associated with $I_{j}$ if $I_{j}\mid N$ and $N_{j}=0$ otherwise for $j=1,...,r$. Then $f$ an  is invertible function \footnote{The proof of this is similar to the proof of Lemma \ref{same-bino} (part $5$).} with inverse given by $(A_{j})_{j=1}^{r}\mapsto \oplus_{j=1}^{r}A_{j}$.

\item Note that 
$S(\Fn)=\left[\bigcup_{B\in T}\bigcup_{(t_{j})_{j}\in D(B)} \{C\leq \Fn \, : \,  C\cong \oplus_{j\in B}t_{j}I_{j} \}\right]\bigcup \{ \textbf{0}\}$ and so $|S(\Fn)|=\left[ \sum_{B\in T}\sum_{(t_{j})_{j}\in D(B)}\left(  \prod_{j\in B  } \binom{n_{j}I_{j}}{t_{j}I_{j}}_{q}  \right)\right] + 1$ (by Corollary  \ref{bcoef-mult}). 

Note that $|S(H_{j})|=1+ \binom{n_{j}I_{j}}{I_{j}}_{q}+ \binom{n_{j}I_{j}}{2I_{j}}_{q}+\cdots +\binom{n_{j}I_{j}}{n_{j}I_{j}}_{q} =\sum_{t_{j}=1}^{n_{j}}\binom{n_{j}I_{j}}{t_{j}I_{j}}_{q} + 1$ for $j=1,...,r$, and thus $\prod_{j=1}^{r}|S(H_{j})|=\prod_{j=1}^{r}\left( \sum_{t_{j}=1 }^{n_{j}} \binom{n_{j}I_{j}}{t_{j}I_{j}}_{q} + 1 \right)$.

\end{enumerate}

\end{proof}


\subsection{Counting $1$-generator $G$-invariant codes}

Let $G\leq Aut_{\F}(\Fn)$. If $\textbf{0}\neq C\subseteq \Fn$ is a cyclic $\F[G]$-submodule, then it will be said that $C$ is a $1$-\textbf{generator} $G$\textbf{-invariant code}. In this section we use the Gaussian binomial coefficient to count all the $1$-generator $G$-invariant codes in $\Fn$. In \cite{QC-1}, Séguin discussed about the enumeration of 1-generator quasi-cyclic codes when $1=(q,m)=(|q|_{m}, n)$, where $|q|_{m}$ is the order of $q$ module $m$. Later, in \cite{QC-2}, J. Pei and X. Zhang solved that problem without requiring that $(|q|_{m}, n)=1$. However, in both cases, they developed their results using properties of polynomial rings over finite fields. As 1-generator quasi-cyclic codes are $1$-generator $G$-invariant codes, our approach is more general.

\begin{lem}\label{count1}
Let $G\leq Aut_{\F}(\Fn)$ be such that $A=\F[G]$ is semisimple. Let  $\{I_{j}\, : \, j=1,...,r\}$ be a basic set of ideals for $\Fn$. Let  $ n_{j}$ and $k_{j}$  be the  multiplicity  of $I_{j}$ in $\Fn$ and $A$,  respectively,  for $j=1,...,r$. Let   $l_{j}=min\{n_{j}, k_{j}\}$ and $h_{j}= dim_{\F}(I_{j})$ for $j=1,...,r$. Let  $Y=\oplus_{j=1}^{r}l_{j}I_{j}$, $X=\{C\subseteq \Fn \; : \; C \, \text{is\, a \, 1-generator\,  } G-\text{invariant \, code} \}$, and $D(B)=\prod_{j\in B}\{1,..., l_{j}\}$ for all $B\in T:= 2^{\{1,...,r\}}-\{\emptyset\}$. Then,

\begin{enumerate}
\item $C\in X$ if and only if $C\mid Y$ and $C\neq \textbf{0}$.

\item $|X|=\sum_{B\in T}\sum_{(t_{j})_{j}\in D(B)} \left( \prod_{j\in B}\binom{n_{j}I_{j}}{t_{j}I_{j}}_{q}\right)=\left[\prod_{j=1}^{r}\left( \sum_{t_{j}=1}^{l_{j}} \binom{n_{j}I_{j}}{t_{j}I_{j}}_{q} + 1 \right)\right]$ $ -1$.

\item If $l_{j}=1$ for $j=1,...,r$, then $|X|=\sum_{B\in T}\left( \prod_{j\in B}\binom{n_{j}I_{j}}{I_{j}}_{q} \right)=$ \linebreak $\left[\prod_{j=1}^{r} \left( \binom{n_{j}I_{j}}{I_{j}}_{q} + 1 \right)\right] - 1$. Moreover, if $k_{j}=1$ for $j=1,...,r$, then  $|X|=\sum_{B\in T}\left( \prod_{j\in B}\frac{q^{n_{j}h_{j}}-1}{q^{h_{j}}-1}\right)=\left[ \prod_{j=1}^{r} \left( \frac{q^{n_{j}h_{j}}-1}{q^{h_{j}}-1}  + 1 \right) \right] -1 $.

\end{enumerate}  

\end{lem}

\begin{proof}

\begin{enumerate}

\item $C\in X$ if and only if $C$ is a cyclic $A$-module and $C\neq \textbf{0}$, which occurs if and only if $C \mid A$ (because $A$ is semisimple) and $C\neq \textbf{0}$. Thus $C\in X$ if and only if $C\mid Y$ and $C\neq \textbf{0}$.

\item  By part $1$, $X=\bigcup_{B\in T}\bigcup_{(t_{j})_{j}\in D(B)} \{C\leq \Fn \, : \, C\cong \oplus_{j\in B}t_{j}I_{j} \}$. Thus

\begin{eqnarray*}
|X|&=&\sum_{B\in T}\sum_{(t_{j})_{j}\in D(B)}|\{C\leq \Fn \, : \, C\cong \oplus_{j\in B}t_{j}I_{j} \}|\\
   &=& \sum_{B\in T}\sum_{(t_{j})_{j}\in D(B)}\left(\binom{\Fn}{\oplus_{j\in B}t_{j}I_{j}}_{q}\right) \\
   &=&\sum_{B\in T}\sum_{(t_{j})_{j}\in D(B)}\left(\prod_{j \in B}\binom{n_{j}I_{j}}{t_{j}I_{j}}_{q}\right), 
\end{eqnarray*} 

where the last equality is by Corollary \ref{bcoef-mult}. On the other hand, by part $1$, $X=\{\oplus_{j=1}^{r}U_{j}\, : \, U_{j}\leq \Fn \, \wedge \, U_{j}\mid l_{j}I_{j}\}- \{ \textbf{0}\}$. Let $f:X \rightarrow \prod_{j=1}^{r} \{Z\leq \Fn \, :  \, Z\mid l_{j}I_{j}  \}-\{\textbf{0}\}$ be given by $f(L)=(L_{j})_{j=1}^{r}$, where $L_{j}$ is the homogeneous component of $L$ associated with $I_{j}$ if $I_{j}\mid L$ and $0$ otherwise. By similar arguments to the given in the proof of Lemma \ref{same-bino} (part $5$), $f$ is a bijection. So $|X|=\prod_{j=1}^{r} |\{Z\leq \Fn \, : \, Z\mid l_{j}I_{j}  \}| - 1=\left[\prod_{j=1}^{r}\left( \sum_{t_{j}=1 }^{l_{j}} \binom{l_{j}I_{j}}{t_{j}I_{j}}_{q} + 1 \right)\right] - 1$.

\item  If  $l_{j}=1$ for $j=1,...,r$, then $D(B)=\{(1,...,1)\}$ for all $B\in T$. Hence, by part $2$,
\begin{eqnarray*}
|X|&=&\sum_{B\in T} \left( \prod_{j\in B}\binom{n_{j}I_{j}}{I_{j}}_{q} \right) \\
   &=& \left[\prod_{j=1}^{r}\left( \binom{n_{j}I_{j}}{I_{j}}_{q} + 1 \right)\right] - 1.
\end{eqnarray*}

Moreover, if $k_{j}=1$ for $j=1,...,r$, then $l_{j}=1$ for $j=1,...,r$. So, by part $2$ and Corollary \ref{simp-part-cor},

\begin{eqnarray*}
|X|&=& \sum_{B\in T}\left( \prod_{j\in B}\frac{q^{n_{j}h_{j}}-1}{q^{h_{j}}-1}\right)\\
   &=&\left[\prod_{j=1}^{r}\left( \frac{q^{n_{j}h_{j}}-1}{q^{h_{j}}-1} + 1 \right)\right]-1 .\\ 
\end{eqnarray*}

\end{enumerate}

\end{proof}

\section{Computing sum of $\F[G]$-submodules}\label{ssix}

Let  $M$ be a finite $A$-module. Due that the submodules of $M$ must be direct sums of simple submodules, these can be computed by taking all the possible sums of simple submodules of $M$. However, if the sums of these simple submodules are not carefully done, the amount of work could increase considerably because every submodule of $M$ may be expressed in many  different ways as a direct sum of simple submodules of $M$. The following result provides a partial solution to that problem.

\begin{lem}\label{algo}
Let $I\leq A$ be a minimal ideal. Let $M$ be a finite  $A$-module, such that $I\mid M$. Let $H$ be the homogeneous component of $M$ associated with $I$, and $n$ be the multiplicity of $I$ in $M$. Let $\{N_{i}\}_{i\in J}$ with $J=\{1,...,r\}$ be the collection of all $A$-submodules of $M$ isomorphic to $I$, and $\dbinom{J}{k}$ the collection of subsets of $J$ with $k$ elements.

Let $(F, Z, X)$ be given as output of  \textbf{Algorithm} \ref{algo1}. Then the following hold:
 
\begin{enumerate}

\item $F$ contains all $A$-submodules of  $H$.

\item If  $y_{0}, y_{1} \in Z $, then $\sum_{j\in y_{0}} N_{j}\neq \sum_{j\in y_{1}} N_{j}$.

\item  For all $y\in Z$, $\sum_{j\in y} N_{j}=\oplus_{j\in y} N_{j}$.

\end{enumerate}

\begin{algorithm}
  \caption{Sum of simple $A$-modules }\label{algo1}
   \begin{algorithmic}[1] 
    \Function{Sumofsimp}{$\{N_{1},...,N_{r}\},\ n$}

      \State $X=\emptyset$; $J=\{1,...,r\}$ \Comment{Where $\emptyset $ is the empty set.}
      
        \State $Z=\{\{1\},\{2\},...,\{r\}\}$
        \State $F=\{N_{1},...,N_{r}, \textbf{0}\}$
        \For{$k = 2$ to ${n}$}
               
          \For { $y\in \dbinom{J}{k}$ }
            \State $R_{y}=\emptyset$ 
            \If{not $y\in X$}
            
            \State add $y$ to $Z$\label{n1} 
            \State  add $\sum_{j\in y} N_{j}$ to $F$
            \State $count=0$

                 \For {$t\in J-y$} \label{n2}
                  \If{$count <\binom{kI}{I}_{q}-|y|$}\label{n4}\Comment{This conditional can be ignored for the proof of Lemma \ref{algo}. Nevertheless, in practice, its omission would cause the algorithm to run more slowly.}
                    \If{$N_{t}\subset \sum_{j\in y} N_{j}$}
                      \State add $t$ to $R_{y}$
                      \State $count=count+1$
                    \EndIf
                  \Else
                    \State $\textbf{break}$  
                  \EndIf  
                 \EndFor

                 \For{$j=k$ to $ min\{|R_{y}\cup y|, n\}$}  \label{n5}
                   \For {$u \in \dbinom{R_{y}\cup y}{j}-Z $}\label{n6}
                         \State  add $u$ to $X$\label{n7}
                   \EndFor
                 \EndFor 
             \EndIf     
               
             \EndFor
           \EndFor
             \State return $(F,Z, X)$
         \EndFunction
\end{algorithmic}
\end{algorithm}

\end{lem}

\begin{proof}
Let $E:=\bigcup_{t=1}^{n}\binom{J}{t}$. Observe that $E=X\sqcup Z$.

\begin{enumerate}

\item Proceeding by induction on the multiplicity $l$ of $I$ in the $A$-submodules of $H$, it is easy to see that the statement holds for $l=2$, i.e., $F$ contains all $A$-submodules of $H$ that are the sum of two simple $A$-modules.\\
 Suppose that the same is true for $l\lneq k$ with $k$ a positive integer, i.e., $F$ contains all $A$-submodules of $H$ that are the direct sum of $l$ simple $A$-modules with $l\lneq k$.  Let $y\in E$ with $|y|=k$. If $y\in Z$, it is clear that $\sum_{j\in y}N_{j}\in F$. Otherwise, if $y\in X$, then there exists $z\in Z$ with $|z|\leq k$ such that $\sum_{j\in y}N_{j}\subseteq \sum_{i\in z}N_{i}\in F$ (by construction of $X$, Algorithm \ref{algo1}, lines \ref{n5}-\ref{n7}). If $\sum_{j\in y}N_{j}= \sum_{i\in z}N_{i}$, the proof ends. On the other hand, if $\sum_{j\in y}N_{j}\subsetneq \sum_{i\in z}N_{i}$, the multiplicity of $I$ in $\sum_{j\in y}N_{j}$ must be less  than the multiplicity of $I$ in $\sum_{i\in z}N_{i}$, which is at most $|z|\leq k $, and by inductive hypothesis, $\sum_{j\in y}N_{j}$ belongs to $F$.
 
\item Suppose that there exist $y, y' \in Z$ with $y\neq y'$ and $\sum_{i \in y}N_{i}= \sum_{j \in y'}N_{j}$. Without loss of generality, $y$ was added before than $y'$ to $Z$ in line \ref{n1} of Algorithm \ref{algo1}. Then $\sum_{i \in y}N_{i}= \sum_{j \in y'}N_{j}$ implies $N_{j}\subseteq \sum_{i \in y}N_{i}$ for all $j\in y'$. If $R_{y}=\{j \in J-y \, : \, N_{j}\subset \oplus_{i\in y}N_{i}\}$ (this  is the set obtained after executing the loop in line \ref{n2}  of Algorithm \ref{algo1}), then $y'\subseteq R_{y}\cup y $. In addition, $|y'|\in \left[ |y|,  min\{|R_{y}\cup y|, n\} \right]$. However, in the step in which $y$ was added to $Z$, $y'$ did not belong to $Z$ (because $y$ was added before than $y'$ to $Z$). Thus, $y'$ was added to $X$ (by construction of $X$, Algorithm \ref{algo1}, lines \ref{n5}-\ref{n7}) which contradicts that $y'\in Z$.

\item Suppose that there exists $y\in Z$ such that $\sum_{i\in y}N_{i}$ is not a direct sum. Then there exists $l\in E$ with $|l|\lneq|y|$ and $\oplus_{j\in l}N_{j}=\sum_{i\in y}N_{i}$.  If $l\in Z$, then $l$ was added first than $y$ to $Z$ in line \ref{n1} of Algorithm \ref{algo1} (because $|l|\lneq |y|$). If $R_{l}=\{i\in J-l \, : \, N_{i}\subset \oplus_{j\in l}N_{j}\}$ (this  is the set obtained after executing the loop in line \ref{n2}  of Algorithm \ref{algo1}), then $y\subseteq R_{l}\cup l $. In addition, $|y|\in \left[ |l|,  min\{|R_{l}\cup l|, n\} \right]$, but in the step in which $l$ was added to $Z$, $y$ did not belong to $Z$ (because $|l|\lneq |y|$). Thus $y$ was added to $X$ (by construction of $X$, Algorithm \ref{algo1}, lines \ref{n5}-\ref{n7}) which contradicts that $y\in Z$. If $l \in X$, then  there exists $y'\in Z$ with $|y'|\leq|l|$ such that $\oplus_{j\in l}N_{j}\subseteq \sum_{j\in y'}N_{j}$ (by construction of $X$, Algorithm \ref{algo1}, lines \ref{n5}-\ref{n7}). So $\sum_{i\in y}N_{i}=\oplus_{j\in l}N_{j}=\oplus_{j\in y'}N_{j}$, but $y\neq y'$ (because $|y'|\leq|l|\lneq|y|$) and $y,y'\in Z $, which contradicts part $2$.

\end{enumerate}

\end{proof}


\section{Computing the $G$-invariant codes of $\F^n$}\label{sseven}

In this section we provide a method to compute all the $\F[G]$-submodules ($G$-invariant codes) of $\Fn$ for a given $G\leq Aut_{\F}(\Fn)$, when $(|G|,q)=1$. In this case, Theorem \ref{maschke} guarantees that the $\F[G]$-submodules of $\Fn$ are direct sums of simple $\F[G]$-submodules. As the simple submodules of  $\Fn$ are contained in their homogeneous components, we will start by giving ways of computing these last.

\subsection{Additional results for the solution of the invariance problem}

In this subsection, results that help to determine the $G$-invariant codes of $\F^n$ are presented, starting with  Theorem \ref{divmin}, which is a particular case of  \cite[Theorem 54.12]{rep2}.

\begin{teo}[Minimal divisor]\label{divmin}
Let $e\in A$ be a primitive idempotent,  $I=Ae$ be the minimal ideal generated by $e$, and $M$ be a finite $A$-module. Then

\begin{center}
$I\mid M$ if and only if $eM\neq 0$.
\end{center}
\end{teo}

\begin{lem}\label{com-hom}
Let $X$ be a basic set of idempotents for a finite $A$-module $M$. Let $e\in X$ be a primitive idempotent such that  $I=Ae\mid M$. Let $H$ be the homogeneous component associated with $I$ and $n$ be the multiplicity of $I$ in $H$. Then the following hold:

\begin{enumerate}
\item $H\cong nI$ as $A$-modules.

\item $e M \cong n(eAe)$ as $\F$-vector spaces.

\item If $e$ is central, then $eM=H$. 
\end{enumerate} 
\end{lem}

\begin{proof}

\begin{enumerate}

\item It is clear.

\item  Note that $eM\cong Hom_{A}(Ae,M)\cong n Hom_{A}(Ae,Ae)=n End_{A}(Ae)\cong n(eAe)$ as $\F$-vector spaces, where the first and last isomorphisms are by Lemma \ref{iso-ev}, and the second isomorphism is due to additivity of  $Hom_{A}(Ae, \;)$ respect to direct sums and Lemma \ref{schur}. Therefore, $eM\cong n(eAe)$ as $\F$-vector spaces.

\item If $e$ is central, $eM$ is clearly an $A$-submodule of $M$ such that for all $f\in X-\{e\}$, $f(eM)= 0$, because the idempotents in $X$ are orthogonal. This implies that the unique  simple divisor (up to isomorphism) of $eM$ is $I$ (by Theorem \ref{divmin}). Thus $eM\leq H$. On the other hand,  by Lemma \ref{iso-ev},  $eM\cong Hom_{A}(Ae,M)\cong n(eAe)=nI\cong H$ as $\F$-vector space, where the last isomorphism  is by part $1$, and thus $eM=H$.
\end{enumerate}

\end{proof}

\begin{cor}\label{divmin1}
Let  $M$, $H$, $e$, and $n$  be as in Lemma \ref{com-hom}. Then, $n=dim_{\F}(H)/dim_{\F}(Ae)=dim_{\F}(e M)/dim_{\F}(eAe)$.
\end{cor}



There are occasions in which one have a generator element of a minimal ideal that  is not an idempotent, and want to determine if the ideal generated by this element divides a given module.   Corollary \ref{divmin1.5} presents a solution to this problem. 

\begin{cor}\label{divmin1.5}
Let $e\in A$ be a primitive idempotent, $M$ be  a finite $A$-module, and $0\neq f\in I=Ae$ be such that $I\mid M$. Then,

\begin{enumerate}
\item $eM\neq 0$ if and only if $fM\neq 0$.

\item If $e$ and $f$ are central, $fM$ is the homogeneous component of $M$ associated with $I$ and the multiplicity of $I$ in $M$ is $n=dim_{\F}(fM)/dim_{\F}(Af)$.

\end{enumerate}
\end{cor}

\begin{proof}
\begin{enumerate}
\item If $eM=0$, then there exists $b\in A$ such that $fM=(be)M=b(eM)=0$. Similarly, if $fM=0$, then $eM=0$.

\item Let $e$ and $f$ be central, then $eM=fM$ and $eAe=fAf=Af$. So, by Corollary \ref{divmin1}, $n=dim_{\F}(eM)/dim_{\F}(eAe)=dim_{\F}(fM)/dim_{\F}(Af)$ is the multiplicity of $I$ in $M$. Besides, if $H$ is the homogeneous component associated with $I$, $H=eM=fM$ (by Lemma \ref{com-hom} part $3$).
\end{enumerate}
\end{proof}


\begin{lem}\label{mcicli}
Let  $A=\F [G]$, $M$ be  an $A$-module, and $m\in M$. Then $A\cdot m=\langle O(m) \rangle_{\F}$, where $O(m)$ is the orbit of $m$ under the multiplication by elements of $G$.
\end{lem}





Lemma \ref{com-hom} (part $3$) offers an easy way to compute the homogeneous component associated with a minimal ideal generated by a central primitive idempotent. Thereby, it is natural to ask over what happens when this idempotent is not central. Theorem \ref{divmin3} will englobe both cases. However, when possible, it is recommended to use Corollary \ref{com-hom} instead of this theorem because it is easier to be applied.

\begin{teo}\label{divmin3}
Let  $M$ be a finite $A$-module. Let $e\in A$ be a primitive  idempotent such that  $I=Ae\mid M$.  Let $H$ be  the homogeneous component of $M$ associated with $I$, and $n$ be the multiplicity of $I$ in $M$. If $\beta=\{\beta_{1}, \beta_{2},...,\beta_{r}\}$ is a generating set of $M$ as  $\F$-vector space, then $\sum_{j=1}^{r} A(e\beta_{i}) =H$.
\end{teo}

\begin{proof}
Let $O(x)$ denote the orbit of $x$ under the left action of $G$ in $M$  for all $x\in M$. From  Lemma \ref{com-hom},  $eM\cong n(eAe)\leq nI $ (as $\F$-vector spaces). Let $y_{i}\in n(eAe)$ be such that its $i$-$th$ entry is $e$ and $0$ otherwise, for $i=1,...,n$. Then $Y:=\{y_{i}\;| \; i=1,...,n \}\subset n(eAe) $ is such that $nI=\oplus_{i=1}^{n}Ay_{i}$. For any $\phi\in Iso_{A}(H, nI)$, its restriction $\eta$ to $eM\subset H$  belongs to $ Iso_{\F}(eM,n(eAe))$. Thus $\eta^{-1}(Y)=\{z_{i}:=\eta^{-1}(y_{i})\;|\;  i=1,...,n \}$ is a subset of $eM$ such that  $H=\oplus_{i=1}^{n}Az_{i}$ (because $Y$ generates $nI$). On the other hand, $X:=\{e\beta_{1}, e\beta_{2},...,e\beta_{r}\}$ generates $eM$ as vector space. Hence $z_{i}=\sum_{j=1}^{r}c_{ij}e\beta_{j}$ for $i=1,...,n$, and so $gz_{i}=\sum_{j=1}^{r}c_{ij}g(e\beta_{j})$ for $i=1,...,n$ and $g\in G$. So $O(z_{i})\subset \left\langle \bigcup_{j=1}^{r}O(e\beta_{j}) \right\rangle_{\F}=\sum_{j=1}^{r}\left\langle O(e\beta_{j}) \right\rangle_{\F}=\sum_{j=1}^{r} A(e\beta_{j})\subseteq H $ for $i=1,...,n$, where the last of the equalities follows from Lemma \ref{mcicli}. Then $\left\langle O(z_{i})\right\rangle_{\F}= Az_{i}\subset \sum_{j=1}^{r} A(e\beta_{j}) $ (by Lemma \ref{mcicli}) and $\oplus_{j=1}^{r}Az_{i}=H\subseteq \sum_{j=1}^{r} (Ae\beta_{j})$.

\end{proof}

\begin{lem}\label{core}

Let $G=\{g_{i} : i=1,...,k\}\leq Aut_{\F}(\Fn)$ such that $A=\F[G]$ is semisimple. For each $x\in \Fn$, let $O(x)$ denote the orbit of $x$ under the action by evaluation of $G$. Let $I=Ae\leq A$  with $e$ a primitive idempotent. Let $M=A m $ and $N=An $ be cyclic $A$-submodules of $\Fn$. Let $B$ be the matrix whose $i$-th row is given by $g_{i}(m)$, and $C$ be the  $(k+1) \times n$ matrix whose $i$-th row is given by $g_{i}(m)$ if $i=1,...,k$ and $n$ otherwise. Then the following  hold: 

\begin{enumerate}
\item $M\cong I$ if and only if $dim_{\F}(I)=rank(B)$ and $e\cdot O(m)\neq \{ 0 \}$. Moreover, if $e$ is central, the condition $e\cdot O(m)\neq \{ 0 \}$ could be replaced by $e\cdot m\neq 0$.

\item  $N\leq M$ if and only if $rank (C) =rank(B)$. Moreover, If $N,M$ are simple $A$-modules, then $N=M$ if and only if $rank (C) =rank(B)$.

\end{enumerate}

\end{lem}

\begin{proof}

\begin{enumerate}

\item If $M\cong I$, then $I\mid M$ and by Theorem \ref{divmin}, $e\cdot M\neq 0$. So  $e\cdot O(m)\neq \{0\}$ (because $M=\langle O(m)\rangle_{\F}$, by Lemma \ref{mcicli}). Thus $e\cdot O(m)\neq \{0\}$ and $dim_{\F}(I)=dim_{\F}(M)=rank(B)$ (by Lemma \ref{mcicli}). Conversely, if $ \{0\}\neq e\cdot O(m)\subseteq e\cdot M$, Theorem \ref{divmin} implies that $I\mid M$. But $dim_{\F}(I)=rank(B)=dim_{\F}(M)$ (by Lemma \ref{mcicli}), therefore $I\cong M$ as $A$-modules. If $e$ is central, $e(O(m))=0$ if and only if $e(m)=0$.
\item If $N\leq M$, then $n\in M=\langle O(m)\rangle_{\F}$ and thus $rank (C) =rank(B)$.\\
Conversely, if $rank (C) =rank(B)$, then $n\in M=\langle O(m)\rangle_{\F}$ and so $g(n)\in M$ for all $g\in G$. Thus $O(n)\subseteq M$, and by Lemma \ref{mcicli}, $N=An=\langle O(n)\rangle_{\F}\leq M$. Moreover, if $N$ and $M$ are simple $A$-modules, then $N\leq M$ if and only if $N=M$, or equivalently, $rank (C) =rank(B)$.
\end{enumerate}
\end{proof}

\subsection{A method to compute $G$-invariant codes}\label{steps}

In \cite{QC-1}, Séguin describes an algorithm to obtain a unique generator for each $q$-ary $1$-generator $m$-quasi-cyclic code of lenght $n$ when $(q,n)=1$ and  the  factorization of $x^{m}-1$ is the same in $\gf{q^{n}}[x]$ as in $\F[x]$. He uses a natural transformation to translate the problem form $(\F[x]/ (x^m-1))^{n}$ to $\gf{q^{n}}[x]/ (x^m-1)$. Later, in \cite{QC-2}, J. Pei and X. Zhang offer a more general approach to the same question using ideas based on properties of semisimple modules. In this section we present a method to find a unique generating set for every $G$-invariant code of $\Fn$ respect to some subgroup $G$ of $Aut_{\F}(\Fn)$. As quasi-cyclic codes are a particular case of $G$-invariant codes, our approach is more general than the presented Séguin, J. Pei, and X. Zhang.\\
Let $G\leq Aut_{\F}(\Fn)$ such that $A=\F[G]$ is semisimple, $X$ be a  basic  set of idempotents for $\Fn$, and $H_{e}$ the homogeneous component associated with $Ae$ for all $e\in X$. By doing what is indicated in the Steps $1-3$ (presented below) for all $e\in X$, we can obtain all simple $A$-submodules of $\Fn$. Then, by doing what is indicated in Step $4$, all  $A$-submodules ($G$-invariant codes) of $\Fn$ are obtained.\\

\underline{Step 1}: (Computation of homogeneous components).  Determine the homogeneous component $H_{e}$ of $\Fn$, which can be done by using Lemma \ref{com-hom} (part $3$), or Corollary \ref{divmin1.5} (part $2$) if $e$ is central. Otherwise, by using  Theorem \ref{divmin3}.\\

\underline{Step 2}: (Computation of quotient sets). Once  $H_{e}$ is determined, considering that all cyclic $A$-modules, and therefore all simple $A$-modules, are generated as $\F$-vector spaces by the orbit of one generating element (by Lemma \ref{mcicli}), the quotient set $H_{e}/G=\{O(m)\;|\; m\in H_{e}\}$ of the orbits under the action by evaluation of $G$ on $H_{e}$ is determined.\\

\underline{Step 3}: (Determination of a unique generating orbit for every simple $A$-module). Determine  those orbits on $H_{e}/G$ that generate simple $A$-modules and obtain a unique generating orbit for every simple $A$-submodule contained in $H_{e}$. All the orbits in $H_{e}/G$ generate $A$-modules which have $I$ (up to isomorphim) as a unique  simple divisor  (by \cite[Proposition 3.20]{rep3}, part $2$). So, by Lemma \ref{core} (part $1$) and Theorem  \ref{divmin}, we just need to check whether an orbit generates a space with the right dimension $i=dim_{\F}(I)$, and obtain a unique generating orbit for every simple $A$-submodule. A  way to do this is as follows: First, if $A$ is non-commutative, compute all the orbits $O$ of  $L:=\{o\in H_{e}/G\; : \; |o|=min\{|u|\, : \, u\in I/G-\{\{0\}\}\, \}\; \}$ such that $dim_{\F}(\langle O\rangle)=i$. When $A$ is commutative, it is not necessary to compute  $L$. In this case, every orbit different from the orbit of the zero vector will generate a simple $A$-module (by  Lemma \ref{ci-simp}). Second, identify when two orbits in $L$ (when $A$ is non-commutative) or $H_{e}/G$ (when $A$ is commutative) generate the same simple $A$-module by using Lemma \ref{core} (part $2$) to obtain only one generating orbit for every simple $A$-submodule contained in $H_{e}$.\\

\underline{Step 4}: (Computation of direct sums). Once we have all the simple $A$-modules contained in $\Fn $ and the multiplicity of $Ae$ in $H_{e}$ (this last can be obtained by using Corollary \ref{divmin1}),  every $A$-submodule of each homogeneous component can be computed, in an efficient way, by using Algorithm \ref{algo1}. After that,  the $A$-submodules of $\Fn $ can be determined by taking all possible direct sums of these submodules, this time without any worry of wasting resources, i.e., with no risks of getting repetitions. Otherwise, the function presented in the proof of Lemma \ref{cont-inv} (part $1$) would not be a bijection.\\

Let us make some remarks on how to get a generating set for every $G$-invariant code. If we provide an  indexed list of all the simple modules contained in a homogeneous component of $\Fn$,  Algorithm \ref{algo1} gives a collection $Z$ of subsets of the set of indices. This collection satisfies that every $A$-submodule of the homogeneous component can be seen as a sum, indexed  by a unique of its elements, of some of these simple modules. Thus we could obtain a unique  generating set for every $A$-submodule of a homogeneous component of $\Fn$ that has been calculated  by the Algorithm \ref{algo1}. For that, we just need to take a non-zero element in each of the simple $A$-modules that appears in its decomposition. In general, for the $A$-submodules of $\Fn$ ($G$-invariant codes), as they are direct sums of $A$-submodules of the homogeneous components of $\Fn$, we just need to take the  unions of the generating sets of their summands.\\
We have just determined how to compute generating sets for $G$-invariant codes. Nonetheless, when working with a code, it is important to know a basis of it. The obvious way to obtain a basis for a $G$-invariant code is by computing it from a generating set. The following results will show another way to do so.

\begin{teo}\label{basisGinv}
Let $A=\F[G]$ be semisimple, $e$ be a primitive idempotent, $Ax$ a cyclic $A$-module isomorphic to $Ae$. Let $B:=\{b_1 , b_2 ,...,b_k\}$ an $\F$-basis for $Ae$. Then the following hold:

\begin{enumerate}

\item If $ex\neq 0$, then $B':=\{b_{1}x , b_{2}x ,...,b_{k}x\}$ is $\F$-basis for $Ax$.

\item There exists $g\in G$ such that $e(gx)\neq0$. Besides, If $e$ is central $ex\neq0$. 

\end{enumerate}

\end{teo}

\begin{proof}

\begin{enumerate}

\item If $f:Ae\longrightarrow Ax$ is given by $f(ae)=(ae)x$, then $f$ is a  non-zero morphism of $A$-modules (because $ex\neq 0$), and so it is an $A$-isomorphism (by Lemma \ref{schur}). Thus $f(B)=B'$ is an $\F$-basis for $Ax$.

\item Due that $Ax\cong Ae$, there exists $g\in G$ such that $e(gx)\neq 0$ (by Theorem \ref{divmin} and Lemma \ref{mcicli}). If $e$ is central $e(gx)=g(ex)\neq 0$ and thus $ex\neq 0$. 
\end{enumerate}
  
\end{proof}

If $A=\F[G]$ is semisimple, and $M$ is a finite (not necessarily simple) $A$-module, then Theorem \ref{basisGinv} can be applied to compute an $\F$-basis for $M$. If $M=\oplus_{i=1}^{t}Am_{i}$ is a decomposition of $M$ into simple submodules, we just need to know a basic set of idempotents for $M$, and which of the ideal generated by these idempotent is isomorphic to $Am_{i}$ for $i=1,...,t$. In this manner, we can determine an $\F$-basis for every summand $Am_{i}$, and hence we can compute a basis for $M$ just by taking the union of these. 

\begin{cor} \label{basisGinv-cor} Let $M$ be a cyclic $A$-module. Let $M=\oplus_{i=1}^{t}Am_{i}$ be the decomposition of $M$ into simple submodules. If $Am_{i}$ is isomorphic to an ideal generated by a central idempotent for $i=1,...,t$, then $M=A\left(\sum_{i=1}^{t}m_{i}\right)$.
\end{cor}

\section{What to do when a basic set of idempotents is not known}\label{seight}

Observe that the previous ideas work if a basic set of idempotents for $\Fn$ is known. In the following lines an alternative solution is discussed. Let $G\leq Aut_{\F}(\Fn)$ such that $\F[G]$ be semisimple. To find a basic set of primitive idempotents for $\F[G]$,  the results presented  in \cite{idemp1}, \cite{idemp2}, \cite{idemp3}, or \cite{split} might be useful. After having determined a basic set of idempotents for $\F[G]$, by using Lemma \ref{divmin}, a basic set of idempotents for $\Fn$ can be computed. Otherwise, considering that, in theory, the use of primitive idempotents of $\F[G]$  to solve the invariance problem is not strictly necessary, we could work using the following reasoning instead.

\begin{lem}\label{G-inv}
Let $G\leq Aut_{\F}(\Fn)$, $S=\{s_{1},...s_{r}\}$ be a generating set for $G$, $\langle s_{i} \rangle$ be the cyclic group generated by $s_{i}$ for $i=1,...,r$, and $C\subseteq \Fn$ be a code. The following conditions are equivalent:

\begin{enumerate}
\item $C$ is a $G$-invariant code

\item $C\in \bigcap_{i=1}^{r} \{ D\subseteq \Fn \, : \, D \text{ is } \F[\langle s_{i} \rangle]-\text{submodule of } \Fn \}$
\end{enumerate}

\end{lem}

\begin{proof}
 $g(C)=C$ for all $g\in G$ if and only if $s_{i}(C)=C$ for all $i=1,...,r$, which occurs if and only if $C$ is $\F[\langle s_{i} \rangle]$-submodule of $\Fn$ for $i=1,...,r$, or equivalently, $C\in \bigcap_{i=1}^{r} \{ D\leq \F \, \mid \, D \text{ is } \F[\langle s_{i} \rangle]-\text{submodule of } \Fn \}$.
\end{proof}

Thus, one could solve the invariance problem by finding the $\langle s_{i} \rangle$-invariant codes, where $S=\{s_{1},...,s_{r}\}$ is a generating set of $G$. This theoretic result is unpractical thought. However, by combining what is known up to now with Lemma \ref{G-inv}, we could be able to lower the computations. For example, in order to compute the $G$-invariant codes of $\Fn$, we could find  first the $N$-invariant codes for some subgroup $N$ of $G$, such that the idempotents of $\F[N]$ are easier to compute. Then, we could see which of these codes are invariant under the elements of $T=\{t_{i}\,\mid \, i=1,...,u\}$, where $T$ is a set of representatives of $G/N$. With that reasoning, the invariance problem could be solved with a more reasonable effort when a basic set of idempotents for $\Fn$ is not known.

\section{Examples of computations of $G$-invariant codes}\label{snine}

In this section we present some examples that illustrate the process of solving the invariance problem and other important results. We apply the steps $1-4$ presented in section \ref{steps}, and use \textit{SageMath} for doing all the calculations.

\begin{ex}
Consider $\gamma\in Aut_{\gf{2}}(\gf{2}^9)$ given by $\gamma=\sigma^3$, where $\sigma$ is the cyclic shift. Then, $\langle \gamma \rangle\cong C_{3}=\{1,x,x^2 \}$ and  $\gf{2}^9$ is an $\gf{2}[C_{3}]$-module. Let $A=\gf{2}[C_{3}]$. As $(|C_{3}|,2)=1$, $A$ is semisimple (by Theorem \ref{maschke}). Using Corollary \ref{divmin1.5} (part $1$) is easy to see that the minimal ideals $I_{0}:=A(x-1)$ and $I_{1}:=A(x^2+x+1)$ of $A$ divide $\gf{2}^9$. Thus, its simple $A$-submodules  ($C_3$-invariant codes)  are isomorphic to $I_{0}$ and $I_{1}$. As these ideals are not isomorphic, $\{I_{0}, I_{1}\}$ is a basic set of ideals for $\gf{2}^9$.\\
Now we are going to compute the $A$-submodules of the homogeneous components $H_{0}$ and $H_{1}$ of $\gf{2}^9$ associated with $I_{0}$ and $I_{1}$, respectively.\\

\underline{Step 1}: (Computation of homogeneous components).  As both $I_{0}$ and $I_{1}$ divide $\gf{2}^9$, and $A$ is a commutative ring, $H_{0}:=(\gamma + id)(  \gf{2}^9  )$ is the homogeneous component of $\gf{2}^9$ associated with $I_{0}$, and $H_{1}:=( \gamma^2 + \gamma +  id)(  \gf{2}^9  )$ is the homogeneous component associated with $I_{1}$ (by Corollary \ref{divmin1.5}, part $2$). Let $\beta$ be the canonical basis of $\gf{2}^9$. Then,  

\[
\begin{array}{cccc}
(\gamma + id)(\beta)=&\{100100000,&010010000,&001001000,\\
 &                       000100100,&000010010,&000001001,\\
 &                       100000100,&010000010,&001000001\}                  
\end{array}
\]

  and $(\gamma + \gamma^2 + id)(\beta)=\{100100100, 010010010, 001001001\}$ generate $H_{0}$ and $H_{1}$ as $\gf{2}$-vector space, respectively.\\

\underline{Steps 2}: (Computation of quotient sets). Let $H_{0}/\left\langle \gamma \right\rangle$ and $H_{1}/\left\langle \gamma \right\rangle$ be the quotient sets determined by the action by evaluation of $\langle \gamma \rangle$ on $H_{0}$ and $H_{1}$, respectively. Let $O(v)$ the orbit of $v$ for all $v\in \gf{2}^9$. Then, $H_{0}/\left\langle \gamma \right\rangle- \{O(000000000)\}$ is given by


\[
\begin{array}{ccc}
\{ $O(100100000)$,& $O(010010000)$,& $O(110110000)$,\\
   $O(001001000)$,& $O(101101000)$,& $O(011011000)$,\\
   $O(111111000)$,& $O(110010100)$,& $O(010110100)$,\\
   $O(101001100)$,& $O(001101100)$,& $O(111011100)$,\\
   $O(011111100)$,& $O(011001010)$,& $O(111101010)$,\\
   $O(001011010)$,& $O(101111010)$,& $O(111001110)$,\\
   $O(011101110)$,& $O(101011110)$,& $O(001111110)$\},
\end{array}
\]

 and $H_{1}/\left\langle \gamma \right\rangle - \{O(000000000)\}$ is given by

 \[
\begin{array}{ccc}
\{O(111111111),& O(100100100),& O(010010010),\\
 O(110110110),& O(001001001),& O(101101101),\\
  O(011011011)\}.                 
\end{array}
\]

\underline{Steps 3}: (Determination of a unique generating orbit  for every simple $A$-module).   As $A$ is commutative, every orbit different from the orbit of the zero vector in $H_{0}/\left\langle \gamma \right\rangle$ and $H_{1}/\left\langle \gamma \right\rangle$  generates a simple $A$-module (by Lemma \ref{ci-simp}). Moreover, as every orbit in  $H_{0}/\left\langle \gamma \right\rangle$ has size $3$ and $|I_{0}|=4$, every simple $A$-module in $H_{0}$ has a unique generating orbit. Similarly, as every orbit in  $H_{1}/\left\langle \gamma \right\rangle$ has size $1$ and $|I_1|=2$, every simple $A$-module in $H_{1}$ has a unique generating orbit.\\
If we determine a unique generating vector for every  simple $\gf{2}[C_3]$-submodule  of $\gf{2}^9 $ that is isomorphic to $I_{0}$, we get 

\[
\begin{array}{cccc}

$$n_{0}$=100100000$,& $$n_{1}$=010010000$,& $$n_{2}$=110110000$,\\
 $$n_{3}$=001001000$,& $$n_{4}$=101101000$,& $$n_{5}$=011011000$,\\
  $$n_{6}$=111111000$,& $$n_{7}$=110010100$,& $$n_{8}$=010110100$,\\
  $$n_{9}$=101001100$,& $$n_{10}$=001101100$,& $$n_{11}$=111011100$,\\
  $$n_{12}$=011111100$,& $$n_{13}$=011001010$,& $$n_{14}$=111101010$,\\
  $$n_{15}$=001011010$,& $$n_{16}$=101111010$,& $$n_{17}$=111001110$,\\
  $$n_{18}$=011101110$,& $$n_{19}$=101011110$,& $$n_{20}$=001111110$.
\end{array}
\]

If we determine a unique generating vector for every  simple $\gf{2}[C_3]$-submodule  of $\gf{2}^9 $ that is isomorphic to $I_{1}$, we get 

\[
\begin{array}{cccc}
$$m_{0}$=100100100$,& $$m_{1}$=010010010$,& $$m_{2}$=110110110$,\\
 $$m_{3}$=001001001$, & $$m_{4}$=101101101$,& $$m_{5}$=011011011$,\\
 $$m_{6}$=111111111$.                   
\end{array}
\]

Let $N_{i}:=An_{i}$ for $i=0,...,20$, and $M_{j}:=Am_{j}$ for $j=0,...,6$.\\
Note that $h_{0}:=dim_{\gf{2}}(H_{0})=dim_{\gf{2}}(\langle (\gamma + id)(\beta) \rangle)=6$ and $h_{1}:=dim_{\gf{2}}(H_{1})=dim_{\gf{2}}(\langle (\gamma + \gamma^2 + id)(\beta) \rangle)=3$. Then, the multiplicities of $I_{0}$ and $I_{1}$ in $H_{0}$ and $H_{1}$ are, in both cases, equal to $3$ (by Corollary \ref{divmin1}). Thus  $H_{0}\cong 3I_{0}$ and $H_{1}\cong 3I_{1}$. Hence $\binom{3I_{0}}{I_{0}}_{2}=\frac{2^{3\times 2}-1}{2^{2}-1}=21$, $\binom{3I_{1}}{I_{1}}_{2}=\frac{2^{3\times 1}-1}{2^{1}-1}=7$ (by Corollary \ref{simp-part-cor}), which coincides with the calculations we have just made.\\

\underline{Steps 4}: (Computation of direct sums). By  using Algorithm \ref{algo1}  all the $A$-submodules of $H_{0}$ and $H_{1}$  can be computed. Let $(F_{0},Z_{0}, X_{0})$, and $(F_{1},Z_{1},X_{1})$ be the outputs given by Algorithm \ref{algo1} for the inputs $(\{N_{i}\, \mid \, i=0,...,20\},3)$ and $(\{M_{i}\, \mid \, i=0,...,6\},3)$, respectively.\\
Then the $A$-submodules of $H_{0}$ isomorphic to $2I_{0}$ are of the form $\oplus_{j\in l}N_{j}$ with $l\in \binom{Z_{0}}{2} $, and
$\binom{Z_{0}}{2}=\{\left\{0, 1\right\},\left\{0, 3\right\},\left\{0, 5\right\}, \left\{0, 13\right\},\left\{0, 15\right\},\left\{1, 3\right\},$\\
$\left\{1, 9\right\}, \left\{1, 10\right\},\left\{2, 3\right\},\left\{2, 4\right\},\left\{9, 2\right\},\left\{2, 10\right\},\left\{3, 7\right\},\left\{8, 3\right\},$\\ 
$\left\{4, 7\right\},\left\{8, 4\right\},\left\{5, 7\right\},\left\{8, 5\right\},\left\{6, 7\right\},\left\{1, 4\right\},\left\{8, 6\right\}\}$

The $A$-submodules of $H_{1}$ isomorphic to $2I_{1}$ are of the form $\oplus_{j\in l}M_{j}$ with $l\in \binom{Z_{1}}{2} $, and
$\binom{Z_{1}}{2}= \{ \left\{0, 1\right\},\left\{0, 3\right\},\left\{0, 5\right\},\left\{1, 3\right\},\left\{1, 4\right\},\left\{2, 3\right\},\left\{2, 4\right\} \}.$ 

There is only one $A$-submodule of $H_{0}$ ($H_{1}$) isomorphic to $3I_{0}$ ($3I_{1}$) which is the homogeneous component $H_{0}$ ($H_{1}$) itself. This can be considered as  $\oplus_{j\in l}N_{j}$ ($\oplus_{j\in l}M_{j}$) with $l\in \binom{Z_{1}}{3} =\{\{0,1,3\}\}$ ($l\in \binom{Z_{0}}{3}=\{\{0,1,3\}\} $).\\ 
By Lemma \ref{cont-inv} (part $1$), the collection  $W:=\{ U\oplus V\, \mid \, U\in F_{0} \, \wedge \,    V\in F_{1} \}$, is precisely the collection of all the $C_{3}$-invariant codes of $\gf{2}^9$. If we take two different elements $(W_{0}, W_{1}),(T_{0}, T_{1})\in F_{0}\times F_{1}$, then $W_{0}\oplus W_{1}\neq T_{0}\oplus T_{1}$. Otherwise, the function given in the proof of lemma  \ref{cont-inv} (part $1$) would not be a bijection. Thus, we can be sure that when calculating $W$ no element will be computed more than once.\\  
Note that  $\binom{2I_{0}}{I_{0}}_{2}=\frac{2^{2\times 2} -1}{2^{2} -1}=5$ and $\binom{2I_{1}}{I_{1}}_{2}=\frac{2^{2\times 1} -1}{2^{1} -1}=3$. Hence  $\binom{3I_{0}}{2I_{0}}_{2}=\frac{\binom{3I_{0}}{I_{0}}_{2}\left[ \binom{3I_{0}}{I_{0}}_{2} - 1 \right]}{\binom{2I_{0}}{I_{0}}_{2}\left[ \binom{2I_{0}}{I_{0}}_{2}- 1\right]}=\frac{21\times 20}{5\times 4}=21= |\binom{Z_{0}}{2}|$ and $\binom{3I_{1}}{2I_{1}}_{2}=\frac{\binom{3I_{1}}{I_{1}}_{2}\left[ \binom{3I_{1}}{I_{1}}_{2}- 1 \right]}{\binom{2I_{1}}{I_{1}}_{2}\left[ \binom{2I_{1}}{I_{1}}_{2}-1\right]}=\frac{7\times 6}{3\times 2}=7=|\binom{Z_{1}}{2}|$ (by Lemma \ref{bcoef}), which is consistent with our computations.\\
Now we are going to count the total of $C_{3}$-invariant codes and $1$-generator $C_{3}$-invariant codes. Let $S(M):=\{C\subseteq M \, \mid \, C \; \text{is\, a \,}  C_{3}-invariant \, code\}$ for all $C_{3}$-invariant code $M\subseteq \gf{2}^9$. Then, by Lemma \ref{cont-inv} (part $2$),

\begin{eqnarray*}
\prod_{j=0}^{1}|S(H_{j})|&=&\prod_{j=0}^{1}\left[ \sum_{t_{j}\in \{1,...,n_{j}\}} \binom{n_{j}I_{j}}{t_{j}I_{j}}_{2} + 1 \right]\\
                         &=& \left[ \sum_{t_{0}=1}^{3} \binom{3I_{0}}{t_{0}I_{0}}_{2} + 1 \right]\times \left[ \sum_{t_{1}=1}^{3} \binom{3I_{1}}{t_{1}I_{1}}_{2} + 1 \right]\\
                         &=&[(21+21+1)+1]\times[(7+7+1)+1]\\
                         &=&704
\end{eqnarray*}  

and as $W=S(\gf{2}^9)$, $|W|=|S(\gf{2}^9)|= \prod_{j=0}^{1}|S(H_{j})|=704$. Let $X=\{C\subseteq \gf{2}^9 \, \mid \, C \, \text{is\, a \, 1-generator\,  } C_{3}-\text{invariant \, code} \}$. Then, by Lemma \ref{count1} (part $3$),

\begin{eqnarray*}
|X|&=& \left[ \prod_{j=0}^{1}\left(  \frac{q^{n_{j}h_{j}}-1}{q^{h_{j}}-1} + 1\right) \right]-1\\
   &=&\left[ \left(  \frac{2^{3\times 2}-1}{2^{3}-1} + 1\right)\times \left(  \frac{2^{3\times 1}-1}{2^{3}-1} + 1\right) \right] -1\\
   &=&[(21+1)\times(7+1)]-1=175.
\end{eqnarray*}

The $1$-generator $C_{3}$-invariant codes of $\gf{2}^9$ are   those isomorphic to $I_{0}$, $I_{1}$, or $I_{0}\oplus I_{1}$ (by Lemma \ref{count1}, part 1). Thus the elements $n_{i}$, $m_{j}$, and $n_{i}+ m_{j}$ for $i=0,...,20$ and $j=0,...,7$, are generator elements of these codes (by Corollary \ref{basisGinv-cor}). Hence we have just determined a unique generator for every $1$-generator $C_{3}$-invariant code (i.e., $1$-generator $3$-quasi-cyclic code).\\

Note that, if $1\leq l_{0},l_{1} \leq 3$, and we want to determine the submodules of $\gf{2}^9$   isomorphic to $l_{0}I_{0}\oplus l_{1}I_{1}$, we just need to consider all the sum of the modules isomorphic $l_{0}I_{0}$ and $l_{1}I_{1}$.\\


\end{ex}

\begin{ex}\label{ex-nocom}
Let $A=\gf{5}[S_{3}]$. As $(|S_{3}|,5)=1$, $A$ is semisimple (by Theorem \ref{maschke}). Let $\beta$ be the canonical basis of $\gf{5}^9$, and $x, y\in Aut_{\gf{5}}(\gf{5}^9)$ be such that

 $$[x]_{\beta}= \begin{bmatrix}
0& 3& 4& 4& 0& 3& 3& 2& 1\\
4& 1& 0& 2& 0& 3& 1& 3& 2\\
0& 1& 2& 4& 4& 2& 1& 2& 2\\
3& 0& 2& 3& 3& 0& 2& 2& 3\\
2& 0& 3& 3& 4& 3& 3& 2& 0\\
2& 0& 1& 4& 2& 0& 1& 4& 0\\
0& 4& 4& 4& 3& 2& 1& 1& 0\\
2& 1& 1& 4& 1& 3& 3& 2& 3\\
1& 1& 2& 4& 1& 1& 0& 2& 2
\end{bmatrix}$$

and

$$[y]_{\beta}= \begin{bmatrix}
0& 2& 0& 0& 1& 2& 0& 0& 0\\
4& 2& 0& 0& 4& 3& 0& 0& 0\\
2& 3& 4& 1& 0& 4& 1& 1& 1\\
1& 1& 0& 2& 2& 3& 0& 3& 4\\
4& 3& 0& 1& 1& 4& 0& 3& 4\\
2& 4& 0& 2& 3& 0& 0& 1& 3\\
1& 3& 0& 3& 4& 4& 1& 3& 0\\
0& 1& 0& 3& 2& 0& 0& 4& 0\\
0& 2& 0& 0& 1& 4& 0& 2& 0\\
\end{bmatrix}$$

Let $G=\langle x,y \rangle$. Then, $f:G\longrightarrow S_{3}= \langle a,b \mid a^3=b^2=1, \, bab=a^2 \rangle =\{1,a,a^2,b,ba,ba^2\}$ given by $x\mapsto a$ and $y\mapsto b$ is an isomorphism. By Example\ref{ex1}, $e_{0}=1+a+a^2+b+ba+ba^2, e_{1}=1+a+a^2 +4b+4ba+4ba^2, e_{2}=2+3a^2 +2b +3ba^2$ form a basic set of idempotents for $\gf{5}[S_{3}]$. Using Theorem \ref{divmin} is easy to see that the minimal ideal $I_{j}:=Ae_{j}$ divides $\gf{5}^9$ for $j=0,1,2$. Thus $\{I_{0},I_{1},I_{2}\}$ is  a basic set of ideals for $\gf{5}^9$. Furthermore, using Corollary \ref{divmin1} is easy to see that the multiplicity of $I_{j}$ in  $\gf{5}^9$ is $j+1$  for $j=0,1,2$.\\ 

\underline{Step 1}: (Computation of homogeneous components). Note that $e_{0}, e_{1}$ are central elements of $A$ and both $I_{0}$, and $I_{1}$ divide $\gf{5}^9$. Then, by Lemma \ref{com-hom}, $H_{0}:=e_{0}\gf{5}^9=(id+x+x^2+y+yx+yx^2)(  \gf{5}^9  )=\left\langle 140442324 \right\rangle_{\gf{5}} $ is the homogeneous component associated with $I_{0}$, and $H_{1}:=e_{1}(  \gf{5}^9  )=(id+x+x^2 +4y+4yx+4yx^2)(  \gf{5}^9  )=\left\langle \{ 10020401,
001112001 \} \right\rangle_{\gf{5}}  $ is the homogeneous component associated with $I_{1}$. On the other hand, by Theorem \ref{divmin3}, the set $\{e_{2}v\, \mid \, v\in \beta\}= \{ 140111222, 142330020, 322033442,
414332434, 322021431,$ 
$ 323102302, 232403203, 231322332,411020204\}$ generates the homogeneous component $H_{2}$ associated with $I_{2}$, as $A$-module.\\

\underline{Steps 2}: (Computation of quotient sets). Let $H_{i}/G$ be the quotient set determined by the action by evaluation of $G$ on $H_{i}$ for $i=0,1,2$. Let $O(v)$ be the orbit of $v$ for all $v\in \gf{5}^9$. Then $H_{0}/G -\{O(000000000)\}$ is given by 

\[
\begin{array}{ccc}
\{  O(410113231) & O(140442324),& O(230334143),\\
 O(320221412)\}
\end{array}
\]

 and $H_{1}/G-\{O(000000000)\}$ is given by
 \[
 \begin{array}{ccc}
\{ O(442204101),& O(110020401),& O(220040302),\\
 O(001112001),& O(111132402),& O(221102303),\\
 O(331122204), & O(441142100), & O(002224002),\\
 O(112244403), & O(222214304), & O(332234200)\}              
\end{array}
\]
 

As $|H_{2}/G|=2667$, we prefer not to write it explicitly. Contrary to the previous example,   not all the orbits in $H_{2}/G$ generate simple $A$-submodules.\\

\underline{Steps 3}: (Determination of a unique generating orbit  for every simple $A$-module). As the multiplicity of $I_{0}$ and $I_{1}$ in $A$ is $1$, every orbit in $H_{0}/G$ and $H_{1}/G$ generates a simple $A$-module (by Lemma \ref{ci-simp}). However, in none of the cases there exists a simple $A$-module with a unique generating orbit. On the other hand, in $H_{2}/G$ there are orbits that do not generate simple $A$-modules, these are precisely those that generate $A$-modules isomorphic to $2I_{2}$. For example, $O(100000032)=\{100000032, 130021232,$ $320034341, 042311021, 012314343,$ 
$ 001430241\}$. We use the instructions given in Step 3 (Section \ref{steps}) to get a unique generating orbit of $H_{i}/G$ for each simple $A$-submodule of $H_{i}$ for $i=0,1,2$. When determining a unique generating vector for the  simple $A$-submodules  of $\gf{5}^9 $ isomorphic to $I_{0}$, we get the vector $l_{0}=140442324$. When  doing the same for the  simple $A$-submodules  of $\gf{5}^9 $ isomorphic to $I_{1}$, we get

\[
\begin{array}{ccc}
 $$m_{0}$=110020401$,& $$m_{1}$=001112001$,& $$m_{2}$=111132402$,\\
 $$m_{3}$=221102303$,& $$m_{4}$=331122204$,& $$m_{5}$=441142100$,
\end{array}
\]

and when  doing it for the  simple $A$-submodules  of $\gf{5}^9 $ isomorphic to $I_{2}$, we get

\[
\begin{array}{cccc}
$$n_{0}$=221100102$,& $$n_{1}$=103100210$,& $$n_{2}$=001331412$,\\
 $$n_{3}$=323311232$,& $$n_{4}$=234001133$,& $$n_{5}$=131401423$,\\
  $$n_{6}$=221011401$,& $$n_{7}$=342101443$,& $$n_{8}$=413201013$,\\
  $$n_{9}$=213021224$,& $$n_{10}$=412101120$,& $$n_{11}$=144011334$,\\
  $$n_{12}$=411001232$,& $$n_{13}$=430201003$,& $$n_{14}$=234301142$,\\
  $$n_{15}$=042010323$,& $$n_{16}$=234141023$,& $$n_{17}$=234110040$,\\
  $$n_{18}$=414301401$,& $$n_{19}$=140111222$,& $$n_{20}$=044411334$,\\
  $$n_{21}$=001100443$,& $$n_{22}$=320321100$,& $$n_{23}$=310110242$,\\
  $$n_{24}$=400401244$,& $$n_{25}$=013201040$,& $$n_{26}$=324111111$,\\
  $$n_{27}$=121001341$,& $$n_{28}$=033301213$,& $$n_{29}$=020311031$,\\
  $$n_{30}$=322021431$.
\end{array}
\]

 Let $L_{0}:= Al_{0}$. Let $M_{i}:=Am_{i}$ and $N_{j}:=Am_{j}$ for $i=0,...,5$ and $j=0,...,30$. $I_{j}$ has multiplicity $j+1$ in $\gf{5}^9$ for $j=0,1,2$. So $H_{0}\cong I_{0}$, $H_{1}\cong 2I_{1}$, and $H_{2}\cong 3I_{2}$. As the multiplicity of $I_{1}$ in $A$ is $1$, $\binom{2I_{1}}{I_{1}}_{5}=\frac{5^{1\times 2}-1}{5^{5}-1}=6$ (by Corollary \ref{simp-part-cor}), which coincides with our calculation. Observe that, up to now; we do not have a formula to calculate $\binom{3I_{2}}{I_{2}}_{5}=31$.\\

\underline{Steps 4}: (Computation of direct sums).  By using Algorithm \ref{algo1},  we can compute all the $A$-submodules of  $H_{2}$. $H_{1}$ only has one non-simple $A$-submodule, which is $H_{1}$ itself, and $H_{0}$ does not have non-simple $A$-submodules, because $H_{0}$ is simple. Let $(F_{2},Z_{2},X_{2})$ be the output given by Algorithm \ref{algo1} for the input $(\{N_{i}\, \mid \, i=0,...,30\},3)$.\\
The $A$-submodules of $H_{1}$ isomorphic to $2 I_{2}$ are of the form $\oplus_{j\in l}N_{j}$ with $l\in \binom{Z_{2}}{2} $, and 

$\binom{Z_{2}}{2}=\{
\left\{0, 1\right\}, \left\{0, 6\right\},\left\{0, 8\right\},\left\{0, 10\right\}, \left\{0, 11\right\},\left\{0, 14\right\},$\\
$\left\{1, 6\right\},\left\{1, 7\right\},\left\{8, 1\right\},\left\{1, 11\right\},\left\{1, 15\right\},\left\{2, 6\right\},\left\{2, 7\right\},\left\{9, 2\right\}$\\ 
$\left\{2, 13\right\},\left\{2, 15\right\},\left\{3, 6\right\},
\left\{3, 7\right\},\left\{9, 3\right\},\left\{10, 3\right\},\left\{11, 3\right\},$\\
$\left\{4, 6\right\}, \left\{4, 7\right\},\left\{8, 4\right\},\left\{9, 4\right\},\left\{12, 4\right\},\left\{5, 6\right\},\left\{5, 7\right\},\left\{8, 5\right\},$\\
$\left\{9, 5\right\}, \left\{12, 5\right\} \}.$

There is only one $A$-submodule of $H_{2}$ isomorphic to $3I_{2}$, which is $H_{2}$ itself. So $W:=\{ L_{0}\oplus U\oplus V\, \mid \, U\in F_{1} \, \wedge \,    V\in F_{2} \}$ is the collection of all the $S_{3}$-invariant codes of $\gf{5}^9$ (by Lemma \ref{cont-inv}).\\  
By doing some computations, we got that $N_{0},N_{1}, N_{2},N_{3},N_{4},N_{5}$ are all the simple $A$-submodules contained in $N_{0}\oplus N_{1}\cong 2I_{2}$. So  $\binom{2I_{2}}{I_{2}}_{5}=6$. Hence $\binom{3I_{2}}{2I_{2}}_{5}=\frac{\binom{3 I_{2}}{I_{2}}_{5}\left[ \binom{3 I_{2}}{I_{2}}_{5} - 1 \right]}{\binom{2 I_{2}}{I_{2}}_{5}\left[ \binom{2 I_{2}}{I_{2}}_{5} - 1 \right]}=31=|\binom{Z_{2}}{2}|$, which coincides with our  calculations.\\
To finish we offer an example over how to compute a basis for a given $G$-invariant code using Theorem \ref{basisGinv}. Observe that $\{e_{2}, ae_{2}\}$, $\{e_{0}\}$, $\{e_{1}\}$ are basis for $I_{2}$, $I_{0}$, and $I_{1}$, respectively. Consider the $A$-submodule $C:=(N_{0}\oplus N_{1})\oplus M_{0} \oplus L_{0}\cong 2I_{2}\oplus I_{1}\oplus I_{0}$. As $e_{2}(n_{0})$, $e_{2}(xn_{1})$, $e_{1}(m_{0})$, and $e_{0}(l_{0})$ are non-zero, then,  by Theorem \ref{basisGinv} (part $1$),
\begin{eqnarray*}
\gamma &:=&(\{e_{2}\cdot n_{0},ae_{2}\cdot n_{0}\}\cup \{e_{2} \cdot x(n_{1}),ae_{2}\cdot x(n_{1})\})\cup\\
        & & \{e_{1}\cdot m_{0}\} \cup \{e_{0}\cdot l_{0}\}\\
        &=&(\{ 144030301, 334400403\}\cup \{ 414320423, 04313042\\
        & & 2\})\cup \{ 110020401\} \cup \{140442324\}\text{\ is a basis for $C$.}
\end{eqnarray*}
\end{ex}

\section{Conclusion}

The $G$-invariant codes were studied as $\F[G]$-submodules of $\Fn$. The $\F[G]$-isomorphisms between $\Fn$ and $\F[G]\times \cdots \times \F[G] $ ($t$-times, where $t=n/|G|$) were classified. Then, by applying results of semisimple representation theory, a method to solve the invariance problem was developed. The Gaussian binomial coefficient for finite $\F[G]$-modules was introduced and studied. This concept turned out to be useful for counting  all the $G$-invariant codes and $1$-generator $G$-invariant codes.

\section{Acknowledgements}The authors thank to professor Ismael Gutierrez and to MSc Jose Antonio Sosaya for their valuable  comments and suggestions that helped to improve the final version of this work.

\end{document}